\documentclass[journal]{IEEEtran}
\usepackage{graphics,
           psfrag,
           epsfig,
           amsthm,
           cite,
           amssymb,
           url,
           dsfont,
           balance,
           enumerate,
           color,
           setspace
}

\usepackage{multirow}
\usepackage{algorithm} 
\usepackage{algorithmicx}
\usepackage{algpseudocode}

\usepackage{pgfplots}
\usepackage{tikz}
\usetikzlibrary{spy,backgrounds}
\usetikzlibrary{shapes, arrows, positioning}
\usetikzlibrary{calc}

\newcommand{\gettikzxy}[3]{%
  \tikz@scan@one@point\pgfutil@firstofone#1\relax
  \edef#2{\the\pgf@x}%
  \edef#3{\the\pgf@y}%
}

\newtheorem{proposition}{Proposition}

\usepackage{amsmath}%[centertags][tbtags][sumlimits][nosumlimits][intlimits][namelimits][nonamelimits]
\usepackage{ragged2e}
\newcommand{\eref}[1]{(\ref{#1})}
\newcommand{\sref}[1]{Section~\ref{#1}}

\newcommand{\pref}[1]{Proposition~\ref{#1}}
\newcommand{\cref}[1]{Constraint~\ref{#1}}

\newcommand{\tref}[1]{Table~\ref{#1}}
\usepackage{epstopdf}
\newcommand{\ignore}[1]{}

\usepackage{afterpage}

\usepackage{multirow}
\usepackage[labelformat=simple]{subcaption}

\newcommand{\algrule}[1][2.0pt]{\par\vskip.5\baselineskip\hrule height #1\par\vskip.5\baselineskip}

\usepackage[utf8]{inputenc}
\include{Conclusion/conclusion}

\begin{document}
	%\begin{spacing}{1.5}
\IEEEoverridecommandlockouts

%\title{\vspace{-.5cm}Cloud Base-Station Offloading with QoS \\Guarantees using Cross-Layer Network Coding}

%\title{Improving Connectivity  of Multi-RIS-Assisted UAV Networks using RIS Partitioning}

%\title{RIS Partitioning, User Clustering, and Link Selection  for Resilient   Non-Orthogonal Multiple Access UAV Networks}

\title{RIS Partitioning and User Clustering for Resilient   Non-Orthogonal Multiple Access UAV Networks }
 
\author{Mohammed Saif, Member, IEEE, and Shahrokh Valaee,  Fellow, IEEE
	
	% <-this % stops a space
	
	\thanks{Mohammed Saif and Shahrokh Valaee are with the Department of Electrical and Computer Engineering, University of Toronto, Toronto, Canada, 
Email: mohammed.saif@utoronto.ca, valaee@ece.utoronto.ca.

This work was supported in part by funding from the Innovation for Defence Excellence and Security (IDEaS) program from the Department of National Defence (DND).

Mohammed Saif work was supported in part by the Natural Sciences and Engineering Research Council of Canada (NSERC) under the Postdoctoral Fellowship Award.
	}
 
	\vspace{-.35cm}
 }

\maketitle

\begin{abstract}
The integration of reconfigurable intelligent surfaces (RISs) and unmanned aerial vehicles (UAVs) has emerged as a promising solution for enhancing connectivity in future wireless networks. This paper designs well-connected and
resilient UAV networks by deploying and virtually partitioning multiple RISs to create multiple RIS-aided links, focusing on a link-layer perspective. The RIS-aided links are created to connect user equipment (UE) to blocked and reliable UAVs, where multiple UEs can transmit to same UAV via RIS using non-orthogonal multiple access (NOMA), granting access to UEs and maximizing network connectivity. We first derive exact and approximated closed-form  expressions for signal-to-interference plus noise ratio (SINR) based on aligned and non-aligned RIS-aided beams. Then, we propose to formulate the problem of maximizing network connectivity that jointly considers (i) UE NOMA clustering, (ii) RIS-aided link selection, and (ii) virtual RIS partitioning. This problem is a computationally expensive combinatorial optimization. To tackle this problem, a two-step iterative approach, called RIS-aided NOMA, is proposed. In the first step, the UEs are clustered to the RISs according
to their channel gains, while UAVs are associated to those generated clusters based on their reliability, which measures the criticality of UAVs. The second step   optimally partitions the RISs to support each
of the cluster members. In this step, we   derive the closed-form equations for the optimal partitioning of RISs
within the clusters. Simulation results
demonstrate that the proposed RIS-aided NOMA 
yields a gain of $30$\% to $40$\%, respectively, compared to UAV traditional scheme.  The finding emphasizes the potential of integrating RIS with UAV communications as a robust and reliable connectivity solution for future wireless communication systems.

\end{abstract}


\section{Introduction}
\IEEEPARstart{T}{he} total number of user equipments (UEs) is projected to grow by over 48\%, reaching approximately 
$13$ billion within the next five years, according to the Cisco Mobility Report \cite{Cisco}. This massive surge in mobile-connected UEs, coupled with the increasing demand for reliable connectivity in emerging services such as HD video streaming, sensing, and information gathering, continues to drive the need for massive connectivity solutions.
To address these challenges, the adoption of innovative approaches such as reconfigurable intelligent surfaces (RISs) and unmanned aerial vehicle (UAV) communications is critical for future wireless networks. RISs, widely recognized as a revolutionary technology for beyond-5G wireless systems, continue to play a pivotal role in enhancing spectral efficiency and connectivity while  reducing power consumption and hardware costs \cite{RIS_Mohanad,9293155}. Alongside RIS, UAV communications have demonstrated immense potential to support massive connectivity by enabling flexible deployment, dynamic coverage, and reliable air-to-ground communication links. The integration of RISs and UAV communications, leading to RIS-assisted UAV systems, has emerged as a promising solution to further enhance connectivity \cite{saifglobecom_E, saifglobecom},  increase energy efficiency \cite{Javad_globecom2023}, and improve physical layer security \cite{PLS} and localization \cite{M, Ali, Mu} in next-generation wireless networks.

\subsection{Related Works}
%1) \textit{Traditional UAV Communications:} 
The flexible deployment of UAVs enables the establishment of short distance line-of-sight (LoS) communication links to geographically distant UEs, facilitating connectivity for a large number of UEs and reducing the probability of unconnected UEs \cite{UAV_economy, Rui_Zhang_UAV, 3GPP}. In addition, UAV-based solutions are particularly advantageous for applications such as information gathering in rural wireless networks and networks lacking centralized infrastructure \cite{8292633}. Connectivity in UAV-assisted systems can be further improved by increasing the number of UAVs deployed, as demonstrated and discussed in previous works \cite{8292633, saifglobecom_E}. However, UAVs are inherently prone to failures due to their limited energy capacity, leading to network interruptions and outages that significantly impact the information flow from UEs to fusion centers. Furthermore, deploying a large number of UAVs in densely populated urban areas presents challenges, including site constraints and limited space \cite{Rui_Zhang_UAV}. On the other hand, several works have investigated the deployment of more nodes (sensors and backhaul links) to maximize connectivity, e.g., \cite{4786516, 4657335, H}. Therefore, it is imperative to study alternative solutions that can mitigate these challenges while maximizing the benefits of deploying additional UAVs, without additional system complexity or cost.

%2) \textit{RIS-assisted UAV Communications:}  
RIS consists of a large number of low-cost passive elements that can be dynamically adjusted to alter their phases, thus changing the impinging electromagnetic signals \cite{4A, 6A, 10A}. By smartly controlling the behavior of the waves, RISs can improve the quality, coverage, and energy efficiency  of wireless networks while improving connectivity via adding multiple cascaded channels \cite{5A}, called RIS-aided links. Using RISs to improve connectivity of UAV networks has taken its shape in the recent literature. Recent works improve the connectivity of UAV networks via a single RIS \cite{saifglobecom}, where the authors formulate the network connectivity problem as semidefinite programming (SDP) and solve it using the CVX toolbox \cite{CON}. The authors of \cite{saifglobecom_E} extend the results of \cite{saifglobecom} to multiple RISs and solve the problem via SDP and matrix perturbation. In \cite{saifglobecom, saifglobecom_E}, each RIS assists to improve the connectivity of the UAVs networks by adding a single RIS-aided link that connects a UE and a UAV. These studies demonstrate good performance compared to traditional UAV communications without RIS assistance, which serve as baseline schemes in this paper. Furthermore, \cite{aydin} uses the RIS to boost the strength of the signals for resilient wireless networks. 

RIS virtual partitioning is a RIS optimization technique that enhances communication performance  \cite{PLS, saifVTC}. The idea of RIS virtual partitioning is to divide RIS into several virtual sections, with the phase shifts of each partition configured to a specific direction. This allows a single RIS plate to simultaneously reflect one signal into multiple cascaded signals in different directions. Specifically, in \cite{PLS}, the RIS is used to amplify the signal intended for legitimate UEs while attenuating the signal for illegitimate UEs, thereby enhancing the network secrecy through RIS partitioning. Additionally, the authors of \cite{saifVTC, saifTCOM}  consider optimizing RIS partitioning to maximize UAV connectivity for a simple model of one UE per RIS, limiting the ability  of  RIS to aid the links of multiple scheduled UEs in massive networks. %As such, multiple scheduled UEs can be connected to  a UAV via a single RIS using the virtual partitioning technique.  

Non-orthogonal multiple access (NOMA) is a viable solution to enhance massive connectivity by scheduling multiple UEs to the same network resources \cite{NOMA-1, NOMA-2}. The integration of RIS and NOMA leads to more efficient spectral resource utilization, enabling the fulfillment of diverse quality-of-service (QoS) requirements \cite{NOMA-3, NOMA-4}. %Previous studies have analyzed the performance of RIS-aided NOMA systems under ideal hardware conditions \cite{NOMA-5, NOMA-6}, comparing their efficiency with orthogonal multiple access (OMA) counterparts. In contrast, the impact of non-ideal hardware impairments in RIS-aided NOMA systems has been investigated in \cite{NOMA-7}, highlighting potential performance degradation. RIS technology has gained significant attention for its ability to enhance received power disparity at the receiver, a crucial factor for efficient NOMA operation \cite{NOMA-5}.
Specifically, the integration of NOMA and virtual RIS partitioning has been explored in \cite{NOMA, optimal, newnewnew}, where each RIS  dynamically adjusts its phase configuration to enhance the channel gain of individual UEs. This approach aims to maximize NOMA performance by increasing channel gain disparity, which is essential for achieving optimal power-domain multiplexing. In \cite{NOMA}, the authors exploit RIS partitioning to improve the performance of STAR-RIS-aided grant-free NOMA (GF-NOMA) communication system.  The work \cite{optimal} proposes power allocation and RIS partitioning optimization for both downlink and uplink RIS-assisted NOMA network that comprises single BS. None of these recent works considers maximizing the connectivity of UAV networks via RIS deployment, partitioning, and NOMA.  This paper considers a scenario for enhancing the connectivity of UAV networks by proposing novel RIS-aided link selection and RIS partitioning approach for uplink multi-RIS-assisted UAV NOMA networks. As such, we connect  multiple scheduled UEs to the most reliable UAVs via the RISs. We define the reliability of UAVs based on their impact on connectivity if they are removed (i.e., failed) from the network. The least reliable UAV is the one that, when removed, decomposes the network into sub-networks.

%This paper considers a more general scenario for enhancing the connectivity of UAV networks by proposing novel RIS-aided link selection, RIS partitioning and placement approaches for uplink multi-RIS-assisted UAV networks. As such, we connect  

\subsection{Motivations and Contributions}
Next-generation wireless networks that require high connectivity and reliability
are inherently vulnerable to interruptions and outages, which need a backup
mechanism at the link-layer to ensure resiliency. This work uses RISs as a resilience
mechanism to counteract such outages—an essential feature in applications such as packet delivery and information
gathering in UAV-assisted networks. %In these scenarios, static UEs remain connected to UAVs via RISs until they complete their transmission, thereby minimizing the need for frequent RIS partitioning.
Hence, we focus on a practical scenario involving a specific target area where a set of UAVs hover at a fixed altitude, connecting with massive UEs, with a few RISs strategically deployed to assist communication between desired UEs and reliable UAVs. The communications between UAVs enhance the connectivity of UEs-RISs-UAVs by providing link redundancy and routing, enabling cooperation, and improving network flexibility.

Motivated by the above, this work  designs a link-layer well-connected and resilient RIS-aided NOMA communication system  by deploying multiple RIS-aided links that connect desired UEs to reliable UAVs using NOMA.  The major contributions of this work can be summarized as
\begin{itemize}
\item We first derive exact and approximated closed-form  expressions for signal-to-interference plus noise ratio (SINR) based on aligned and non-aligned RIS-aided beams. Then, we propose to formulate the problem of maximizing the network connectivity that jointly considers (i) UE NOMA clustering, (ii) RIS-aided link selection, and (ii) virtual RIS partitioning. This problem is a computationally expensive combinatorial optimization.

\item To tackle this problem, a two-step iterative approach, called \textit{RIS-aided NOMA}, is proposed. In the first step, the UEs are clustered to the RISs according
to their channel gains, while UAVs are associated to those generated clusters based on their reliability, which measures the criticality of UAVs. The proposed UE clustering solution is intended to work with the max-sum link quality, where we utilize a linear sum assignment
(LSA) for UE clustering and UAV assignment. The second step   optimally partitions the RISs to support each of the cluster members. In this step, we   derive the closed-form equations for the optimal partitioning of RISs
within the clusters.

\item Simulation results
demonstrate that the proposed RIS-aided NOMA 
yields a gain of $30$\% to $40$\%, respectively, compared to UAV traditional scheme without RISs.  The finding emphasizes the potential of integrating RIS with UAV communications as a robust and reliable connectivity solution for future wireless communication systems.

\end{itemize}

Next, Section \ref{S} presents the system model. Section \ref{PF} delves into the problem formulation and presents the proposed methodology. Section \ref{Link} proposes the UE NOMA clustering and UAV selection. Section IV presents the closed-form optimal RIS partitioning. Section \ref{NR} presents the numerical results obtained by simulations. Finally, Section \ref{C} concludes the paper.

\section{System Model}\label{S}

\subsection{Network Model}\label{NM} 
Consider the schematic shown in Fig. \ref{fig1} with multiple single-antenna UAVs and UEs, and 
multiple RISs with $K$ passive reflecting elements. We represent the sets of UAVs, UEs, and RISs as $\mathcal A$, $\mathcal U$, and $\mathcal R$, respectively. Each uniform planar array (UPA) RIS is equipped with $K=K_h \times K_v$ elements where $K_h$ and $K_v$ separately
denotes the sizes along the horizontal and vertical dimensions
of  RIS, respectively. The
RISs are essential for maintaining connectivity between
UEs and UAVs since the direct UE-UAV links might be obstructed.  Assuming a dense urban scenario, where direct links between the UEs and some UAVs are blocked, RISs can significantly aid in establishing reliable communication with the blocked UAVs. In Fig.~\ref{fig1}, for example, if UE$_{1}$ is not connected, it becomes completely blocked due to the absence of direct links to the UAVs, making the network unconnected (zero connectivity).  

In NOMA schemes, the disparity in power reception is a critical factor influencing the overall performance gain over OMA. Similar to \cite{NOMA, optimal}, we overcome this by partitioning the RISs to serve multiple UEs, where each partition is configured to enhance the quality of its corresponding RIS-aided link. Consequently, each UAV receives signals with coherently aligned phases from its own RIS partition and non-coherently aligned phases from other RIS partitions. For practical implementation, we consider scheduling a small number of UEs per RIS using NOMA due to two  reasons. First, serving multiple UEs through NOMA significantly increases the processing complexity of successive interference cancellation (SIC), which grows nonlinearly with the number of UEs \cite{NOMA}. Second, scheduling more UEs to RIS using NOMA does not substantially improve network connectivity, as each RIS would need to be divided into multiple partitions, resulting in weaker RIS-aided links. This observation is further supported by the numerical results presented in Section \ref{NR}.
Therefore, this study focuses on deploying multiple strong RIS-aided links to optimize connectivity and performance, while practically implementing NOMA.

%For the purpose of practicality, we consider only two NOMA UEs per RIS in this paper due to the following two reasons. First, it is important to acknowledge that serving multiple NOMA UEs per RIS can pose significant challenges, as the complexity of the SIC process escalates exponentially with an increased number of UEs, as previously reported in the literature such as [39]. Second,  

  %This improves the overall network connectivity performance of the system.

\begin{figure}[t!]  
\begin{center}
\includegraphics[width=0.99\linewidth, draft=false]{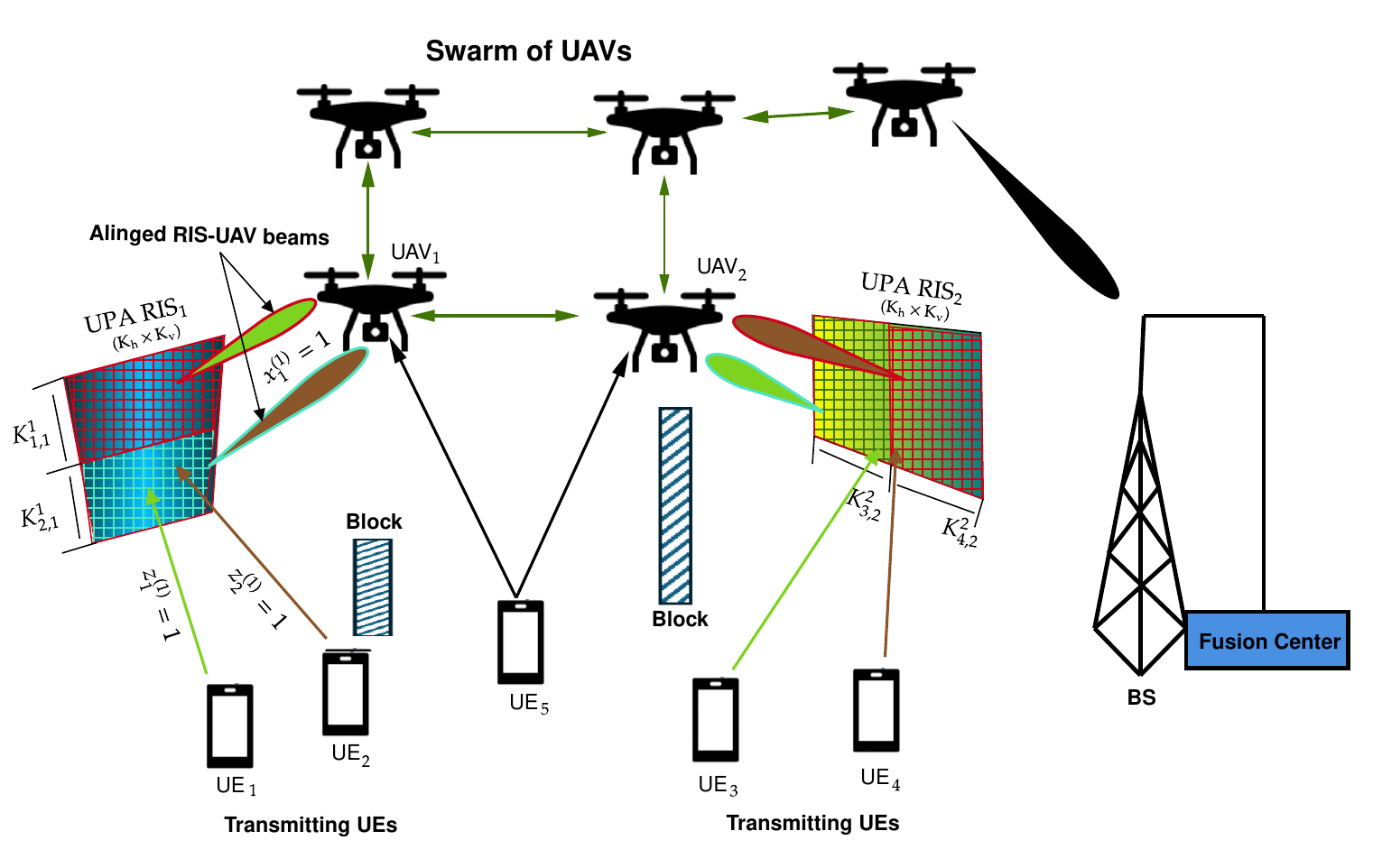}
\caption{RIS-assisted NOMA UAV system with two RISs, each utilizes RIS partitioning technique to aid the signals of two UEs.}
   \label{fig1}
\end{center}
 
\end{figure}

To demonstrate the suitability of this study for massive connectivity, we consider  
a practical dense network scenario where the number of UEs greatly
exceeds the number of available UAVs and RISs, i.e., $U>AR$.  The UEs are grouped into NOMA clusters, denoted by $\mathcal U_r, \forall r \in \mathcal R$. $\mathcal U_r$ also represents the number of partitions in RIS$_r$ and $U_r=|\mathcal U_r|$. The number of NOMA clusters, each with size $U_r$, is equal to the number of RISs $R$, and the terms RIS and cluster are used interchangeably throughout this paper. For simplicity, $\mathcal U_r$ is the associated cluster of RIS$_r$.  Through virtual partition representation, each RIS can connect $U_r$ UEs to one UAV. In this case, the partitions of the RISs  need to be designed to simultaneously boost the signal dedicated to the first  $R$ most reliable UAVs.  The set of the possible transmitting UEs, denoted by $\mathcal U_t$, is at most $\mathcal U_rR$. Fig. \ref{fig1} illustrates that each RIS is partitioned into two portions, where each portion is configured to passively beamform the signals of UE$_{u \in \mathcal U_t}$ to the corresponding UAVs (i.e., $\mathcal U_{r=1}=\{\text{UE}_1, \text{UE}_{2}\}, \mathcal U_{r=2}=\{\text{UE}_3, \text{UE}_{4}\}$ and $U_{r=1}=U_{r=2}=2$).   Low-reliability UAVs are the most critical, as their failure—due to issues such as low battery, hardware malfunctions, or software problems—can cause severe connectivity degradation. Hence, it is essential to avoid relying on these UAVs. The UAV reliability metric, as defined in Section \ref{NRE}, ensures that signals from UEs are transmitted to the most reliable UAVs,   minimizing the risk of connectivity loss and enhancing network resilience.

If RIS$_r$ is assigned to serve cluster $\mathcal U_r$, the number
of elements allocated to aid the signal of UE$_u$ (i.e., $u \in \mathcal U_r$)  to UAV$_a$  is denoted by $K^a_{u,r}=\lceil\alpha^{a}_{u,r}K\rceil$, where $\alpha^{a}_{u,r} \in [0, 1]$ is the RIS allocation factor such that $\sum_{u \in \mathcal U_r} \alpha^{a}_{u,r}  \leq 1$ and $\sum_{u \in \mathcal U_r}K^{a}_{u,r}=K$ ensure  the physically available  number of RIS elements. 
Furthermore, the reflection coefficient matrix of  RIS$_r$ is modeled as $\mathbf{\Theta}_{r} =\text{diag}(J)$, where $J_{r,k}= e^{j\theta^{(k)}_{r}}$ and $\theta^{(k)}_{r} \in [0, 2\pi)$ is the phase shift of the $k$-th element of RIS$_r$, $\forall k$. %In this case, an effective UE-RIS-UAV assignment, RIS partitioning, and RIS positioning approach is crucial to create multiple cascaded links and maximize network connectivity subject to RIS elements, reliability of UAVs, and their QoS.

\subsection{Channel Models and Signal Reception} \label{CM} 
\textbf{1) Channel Models:} This work adopts the Nakagami-$f$ fading model \cite{NOMA, 9A} and the channel state information (CSI) is assumed to be available by
using advanced model-based or data-driven channel estimation techniques, e.g., \cite{CSI-1}. Moreover, the
UEs and UAVs are assumed to have limited mobility and experience slow-fading.  The $\text{UE}_u \rightarrow \text{RIS}_r$ channel  is represented by $\mathbf g_{u,r}=[\mathbf g^{(1)}_{u,r}, \ldots, \mathbf g^{(U_r)}_{u,r}]$, consisting of the
channels from UE$_u$ to all $U_r$ partitions of RIS$_r$. In other terms, $\mathbf g_{u,r}=[g^{(1)}_{u,r}, \ldots,  g^{(k)}_{u,r}, \ldots, g^{(K)}_{u,r}]$, where $g^{(k)}_{u,r}=|g^{(k)}_{u,r}|e^{-j\phi^{(k)}_{u,r}}$  denotes the channel coefficient between UE$_u$ and the $k$-th element of RIS$_r$; $|g^{(k)}_{u,r}|$ is the channel amplitude and $\phi^{(k)}_{u,r}$ is the channel phase. The corresponding large-scale path loss for $\text{UE}_u \rightarrow \text{RIS}_r$ is denoted by  $\Tilde{G}_{u,r}$, which is given by $\Tilde{G}_{u,r}=d_{u,r}^{-\frac{\tau}{2}}$ \cite{NOMA, 9A}, where $d_{u,r}$ denotes the corresponding distance of UE-RIS link and $\tau$ is the path-loss coefficient. Likewise, the channel for $\text{RIS}_r \rightarrow \text{UAV}_a$ is represented by $\mathbf h_{r,a}=[\mathbf h^{(1)}_{r,a}, \ldots, \mathbf h^{(U_r)}_{r,a}]$, comprising  
channels from all $U_r$ partitions of RIS$_r$ to UAV$_a$. In other terms, $\mathbf h_{r,a}=[h^{(1)}_{r,a}, \ldots,  h^{(k)}_{r,a}, \ldots, h^{(K)}_{r,a}]$, where $h^{(k)}_{r,a}=|h^{(k)}_{r,a}|e^{-j\psi^{(k)}_{r,a}}$  denotes the channel coefficient between the $k$-th element of RIS$_r$ and  UAV$_a$; $|h^{(k)}_{r,a}|$ is the channel amplitude and $\psi^{(k)}_{r,a}$ is the channel phase.  
The corresponding large-scale path loss for $\text{RIS}_r \rightarrow \text{UAV}_a$ is denoted by  $\Tilde{H}_{r,a}$, which is given by $\Tilde{H}_{r,a}=d_{r,a}^{-\frac{\tau}{2}}$, where $d_{r,a}$ denotes the corresponding distance of RIS-UAV link. Therefore, we consider the cascaded channel from UE$_u$ to UAV$_a$ over the $k$-th element of RIS$_r$ as $h^{(k)}_{u,r,a}= g^{(k)}_{u,r}  h^{(k)}_{r,a}$. Moreover, for the direct links between the UEs and the UAVs, let $h_{u,a}$ and $\Tilde{G}_{u,a}$ denote the small-scale fading coefficient and path-loss for the $\text{UE}_u \rightarrow \text{UAV}_a$ channel, respectively.

%In the depicted Fig.~\ref{fig1}, for $\text{UE}_u \rightarrow \text{UAV}_k$ and $\text{UE}_u \rightarrow \text{UAV}_{k'}$ over RIS$_m$, we configure the phase shift $\theta^{(m)}_n$ of the $n$-th element of RIS$_m$, $n\in \{1, \ldots, N^{(m)}_{u,k}\}$, to align with UAV$_k$ and the phase shift of the $\Tilde{n}$-th element of RIS$_m$, $\Tilde{n}\in \{N^{(m)}_{u,k}+1, \ldots, N\}$, to align with UAV$_{k'}$.

\textbf{2) Signal Reception:} For cluster $\mathcal U_r$
assisted by RIS$_r$, the  signal received at
UAV$_a$ can be written as follows\footnote{The  signal received at UAV$_a$ can also consider the different UEs in clusters $\mathcal U_{r'}, \forall r'\in \mathcal R, r'\neq r$. However, those different UEs in  $\mathcal U_{r'}$ are scheduled to transmit to different UAVs$_{a'}$ via RISs$_{r'}$, thus we ignore them in \eref{YO}.}
\begin{align}\label{YO}
y^{(r)}_a=y_u+ \sum_{u' \in \mathcal U_r, u' \neq u} y_{u'} +\underbrace{\zeta_a}_\textbf{AWGN}, 
\end{align}
where the first term is the received signal from  UE$_u \in \mathcal U_r$, the second term is the received signal from UE$_{u'} \in \mathcal U_r$, and the third term is the  additive
white Gaussian noise (AWGN) power density at UAV$_a$ with $w_a \sim  \mathcal {CN}(0, \sigma^2_{\zeta_a})$. The first term of \eref{YO} can be expressed as follows
\begin{align}\label{Yu}
&y_u=     \bigg(\underbrace{\sqrt{\Tilde{G}_{u,a}}h_{u,a}}_\text{direct link}+ \underbrace{\sqrt{\Tilde{G}_{u,r} \Tilde{H}_{r,a}}\mathbf h_{r,a} \mathbf{\Theta}_{r} \mathbf g_{u,r}}_\text{via RIS$_r$ aligned for UE$_u\in \mathcal U_r$} \nonumber \\&  
+ \underbrace{\sum_{r' \in \mathcal R, r' \neq r}\sqrt{\Tilde{G}_{u,r'} \Tilde{H}_{r',a}}\mathbf h_{r',a} \mathbf{\Theta}_{r'} \mathbf g_{u,r'}}_\text{via RIS$_{r'}$ aligned for $\mathcal U_{r'}$}\bigg)\sqrt{p}s_u,
\end{align} 
where $s_u$ is the signal of UE$_u$ and  $p$ is the transmit power of UE$_u$, which is fixed for all UEs. $y_u$ in \eref{Yu}, which can also be used for $y_{u' }$ by replacing $u'$ with $u$, has the following components
\begin{enumerate}

\item Signal received from the $\text{UE}_u \rightarrow \text{UAV}_a$ direct link, $\sqrt{\Tilde{G}_{u,a}}h_{u,a}$;

\item Signal received from UE$_u$ to UAV$_a$ via RIS$_r$, $\sqrt{\Tilde{G}_{u,r} \Tilde{H}_{r,a}}\mathbf h_{r,a} \mathbf{\Theta}_{r} \mathbf g_{u,r}$, which can be written as 
\begin{align}\label{Yuu} \nonumber
& \sqrt{\Tilde{G}_{u,r} \Tilde{H}_{r,a}}\mathbf h_{r,a} \mathbf{\Theta}_{r} \mathbf g_{u,r}=\underbrace{\sqrt{\Tilde{G}_{u,r} \Tilde{H}_{r,a}}  \sum_{k=1}^{K^{a}_{u,r}}h^{(k)}_{u,r,a}e^{j\theta^{(k)}_{u,r}}}_\text{aligned signal from portion $K^{a}_{u,r}$}\\&   
+\underbrace{\sqrt{\Tilde{G}_{u,r} \Tilde{H}_{r,a}}  \sum_{u' \in \mathcal U_r, u' \neq u}\bigg( \sum_{k'=1}^{K^{a}_{u',r}}h^{(k')}_{u,r,a}e^{j\theta^{(k')}_{u',r}}\bigg)}_\text{non aligned signal from portion $K^{a}_{u',r}$}, 
\end{align}
where $\theta^{(k)}_{u,r}$ is the $k$-th associated phase shift of RIS$_r$ with UE$_u$. 
The terms of \eref{Yuu} can be written as $\sqrt{\Tilde{G}_{u,r} \Tilde{H}_{r,a}}  \sum_{k=1}^{K^{a}_{u,r}}h^{(k)}_{u,r,a}e^{j\theta^{(k)}_{u,r}}=\sqrt{\Tilde{G}_{u,r} \Tilde{H}_{r,a}}  \alpha^{a}_{u,r}\sum_{k=1}^{K}h^{(k)}_{u,r,a}e^{j\theta^{(k)}_{u,r}}$
and $\sqrt{\Tilde{G}_{u,r} \Tilde{H}_{r,a}}  \sum_{k'=1}^{K^{a}_{u',r}}h^{(k')}_{u,r,a}e^{j\theta^{(k')}_{u',r}}=\sqrt{\Tilde{G}_{u,r} \Tilde{H}_{r,a}}  \alpha^{a}_{u',r}\sum_{u' \in \mathcal U_r, u' \neq u}\bigg(\sum_{k=1}^{K}h^{(k)}_{u,r,a}e^{j\theta^{(k)}_{u',r}}\bigg)$.%, where $\alpha_u$ and $\alpha_{u'}$ are RIS$_r$ portions associated to UE$_u$ and UE$_{u'}$, respectively. 
%This approximation will be compared with the exact model in the prospective
%numerical results section.

\item Lastly, signal received from other RISs$_{r' \in \mathcal R}$ to UAV$_a$, $\sqrt{\Tilde{G}_{u,r'} \Tilde{H}_{r',a}}\mathbf h_{r',a} \mathbf{\Theta}_{r'} \mathbf g_{u,r'}$.
\end{enumerate}
It is worth noting that in UAV communications, UAVs are deployed to be far away from each other to avoid any collision. Additionally, this paper considers narrow-beam RIS beamforming to UAVs. Consequently, the non-aligned signals  from the other RISs (i.e., the second term in \eref{YO}) are small, which can be neglected. Furthermore, due to the
difficulty of deriving a closed-form solution of $\alpha^{a}_{u,r}$ and $\alpha^{a}_{u',r}$ from \eref{exact}, we ignore the impact of non-aligned phases of the same RIS.
Nevertheless, this assumption has been numerically evaluated in the results section;
please refer to Section \ref{NR}.

In light of these discussions, \eref{YO} can be rewritten as  
\begin{align}\label{Yfinal}  \nonumber
y^{(r)}_a(\boldsymbol{\alpha}_r)=& \sum_{u\in \mathcal U_r} \bigg(  \bigg[\sqrt{\Tilde{G}_{u,a}}h_{u,a}+   \alpha^{a}_{u,r}\sqrt{\Gamma_u}\sum_{k=1}^{K} h^{(k)}_{u,r,a} e^{j\theta^{(k)}_{u,r}}\bigg]  \\&  \times s_{u}\sqrt{p} \bigg)
+\zeta_a,
\end{align}
where $\boldsymbol{\alpha}_r=[\alpha^{a}_{u,r}, \ldots, \alpha^{a}_{U_r,r}]$ and $\Gamma_u=\Tilde{G}_{u,r} \Tilde{H}_{r,a}$.
Given the fixed transmit
powers of UEs, distinguishing between the UEs receiving signals at the UAVs is
critical. However, the related channel gain variations and the associated RIS partitions of each UE can achieve this difference. In this case, the stronger UE’s signal message gets decoded while the messages of the other UE is treated as interference. Specifically,  at UAV$_a$, the received signal, $y^{(r)}_a(\boldsymbol{\alpha}_r)$, is decoded using
the SIC scheme. In this iterative scheme, the signal from the UE with the strongest received power is decoded first while treating signals from other UEs as interference. Once decoded, this signal is subtracted from the total received signal, allowing the next strongest UE's signal to be extracted. This process continues iteratively until all UEs in the ordered cluster have been successfully decoded.
\subsection{RIS Phase Shift and SINR Formulation}
\textbf{1) RIS Phase Shift:}
In this paper, we consider two different cases of RIS phase shifts.
\begin{itemize}
\item \textit{Optimal Phase Shift:} Assuming the  perfect CSI is available, BS computes the
optimal RIS phase shifts and transmits them to the RISs controller via dedicated wireless   or wired  feedback channels. Thus, the phase adjusting coefficient for UE$_i$
$\theta^{(k)}_{i,r}$ is given by $\theta^{(r)}_{i,k}=\arg[h_{i,a}]-\big(\phi^{(k)}_{i,r}+\psi^{(k)}_{r,a}\big)$, $\text{UE}_i \in \{\text{UE}_u, \ldots, \text{UE}_{U_r}\}, \forall k$ \cite{Ednew}. Specifically, for the $k$-th element of RIS$_r$, given $\arg[h_{i,a}]$, $\phi^{(k)}_{i,r}$, and $\psi^{(k)}_{r,a}$, the  controller can cophase $\theta^{(k)}_{i,r}$ with those angles to nullify their effect, which provides the maximum channel gain for $\text{UE}_i \xrightarrow{\text{RIS}_r} \text{UAV}_{a}$ link.  Thus, the phases of
the link between  UE$_i$ and UAV$_a$ via RIS$_r$ is cophased. We
can then express the overall cascaded link between UE$_i$ ($\text{UE}_i \in \{\text{UE}_u, \ldots, \text{UE}_{U_r}\}$) and  UAV$_a$ via RIS$_r$ as
follows
\begin{align}\label{RISH_i}
H_i= \sum_{k=1}^K h^{(k)}_{i,r,a}=\sum_{k=1}^Kg^{(k)}_{i,r}h^{(k)}_{r,a}e^{j\theta^{(k)}_{i,r}}=\sum_{k=1}^K|g^{(k)}_{i,r}||h^{(k)}_{r,a}|.
\end{align}

\item \textit{Discrete Phase Shifts:} In practical systems, discrete phase-shift adjustment is commonly employed \cite{PLS, 38}, where each RIS element can select from  $2^b$ distinct phase shift values, with $b$  denoting the bit resolution of the phase shifter. In this study, we consider phase shift control using quantized phase shift levels, as described in \cite{38}.  For example,  with $b=4$, each RIS element has $16$  possible phase shift options. In contrast, a low bit resolution, such as  $b=2$, allows only 
$4$ phase shift options, resulting in imperfect phase shifts. Thus, we consider a discrete phase shift at the RISs as $\theta^{(k)}_{i,r}=\bigg\{e^{\frac{j2\pi c}{2^b}}\bigg\}_{c=0}^{2^b-1}, \forall r, i,k$ \cite{38}.
\end{itemize}

\textbf{2) SINR Formulation:}
Given \eref{Yfinal}, the signal-to-interference plus noise ratio (SINR) at UAV$_a$ can be mathematically expressed as
\begin{align}\label{exact}
 \gamma^{(r)}_{u,a}(\boldsymbol{\alpha}_r)=\frac{p \bigg| X_u^{\text{direct}}+  X^{\text{aligned}}_u\bigg|^2}{\sum_{u' \in \mathcal U_r, u' \neq u} \bigg(p \bigg| X_{u'}^{\text{direct}}+  X^{\text{aligned}}_{u'}\bigg|^2\bigg)+\sigma^2_{\zeta_a}}, 
\end{align}
where $X_u^{\text{direct}}=\sqrt{\Tilde{G}_{u,a}}h_{u,a}$, $X^{\text{aligned}}_u=\alpha^{a}_{u,r}\sqrt{\Gamma_u}\sum_{k=1}^{K} h^{(k)}_{u,r,a} e^{j\theta^{(k)}_{u,r}}$, $X^{\text{aligned}}_{u'}=\alpha^{a}_{u',r}\sqrt{\Gamma_{u'}} \sum_{k=1}^{K}h^{(k)}_{u',r,a} e^{j\theta^{(k)}_{u',r}}$. After configuring
the required phase shifts while ignoring the non-aligned RIS portion, the SINR can be written as 
\begin{align}\label{exact_1} 
 \gamma^{(r)}_{u,a}(\boldsymbol{\alpha}_r) 
= \frac{p \bigg| \sqrt{\Tilde{G}_{u,a}}h_{u,a}+  \sqrt{\Gamma_u} \alpha_u H_u \bigg|^2}{\sum\limits_{\substack{u\in \mathcal U_r \\ u'\neq u}} \bigg(p \bigg| \sqrt{\Tilde{G}_{u',a}}h_{u',a}+  \sqrt{\Gamma_{u'}} \alpha_{u'} H_{u'} \bigg|^2 \bigg) +\sigma^2_{\zeta_a}},
\end{align}
where $H_{i \in \mathcal U_r}$, given in \eref{RISH_i}, is the $K$-element double-Nakagami$-f$ that is independent and identically distributed (i.i.d.) random variable (RV) with parameters $f_1$, $f_2$, $\Omega_1$, and $\Omega_2$ for the $g^{(k)}_{u,r}$ and $h^{(k)}_{r,a}$ links, respectively, i.e., the distribution of the product of two RVs following the Nakagami$-f$ distribution with the probability density function (PDF) is given in \cite{PLS}.

Furthermore, to facilitate the analysis for closed-form RIS partition optimization and similar to \cite{NOMA, PLS}, we make an approximation for SINR expression in \eref{exact_1} utilizing the
expected values of $X_u^{\text{direct}}$, $X_{u'}^{\text{direct}}$, $X^{\text{aligned}}_u$ and $X^{\text{aligned}}_{u'}, \forall u, u' \in \mathcal U_r$. Specifically, the expectations can be evaluated as $\text {E} \bigg\{ \bigg|X_u^{\text{direct}}\bigg|^2\bigg\}=\text {E} \bigg\{ \bigg|   \sqrt{\Tilde{G}_{u,a}}h_{u,a}    \bigg\} = \Tilde{G}_{u,a}$, since $h_{u,a} \sim \mathcal {CN}(0,1)$,  $\text {E}\big\{|h_{u,a}|^2\big\}=1$; $\text {E} \bigg\{ \bigg| X_u^{\text{direct}} \bigg|^2\bigg\}=\text {E}\bigg\{\bigg|\sqrt{\Gamma_u} \alpha_u H_u\bigg|^2\bigg\}=\text {E}\bigg\{\bigg|\sqrt{\Gamma_u} \alpha^{a}_{u,r}  \sum_{k=1}^K|g^{(k)}_{u,r}||h^{(k)}_{r,a}|\bigg|^2\bigg\}=\Gamma_u\text {E}\bigg\{\bigg|\alpha^{a}_{u,r}  \sum_{k=1}^K|g^{(k)}_{u,r}||h^{(k)}_{r,a}|\bigg|^2\bigg\}=\Gamma_u(\alpha^{a}_{u,r})^2 K^2\frac{1}{f_1 f_2} \frac{\Gamma(f_1+0.5)^2}{\Gamma(f_1)^2}\frac{\Gamma(f_2+0.5)^2}{\Gamma(f_2)^2}=\Gamma_u(\alpha^{a}_{u,r})^2 K^2m$, where $m=\frac{1}{f_1 f_2} \frac{\Gamma(f_1+0.5)^2}{\Gamma(f_1)^2}\frac{\Gamma(f_2+0.5)^2}{\Gamma(f_2)^2}$ and $\Gamma(.)$ denotes the Gamma function.

By substituting these values in \eref{exact_1}, we have 
\begin{align}\label{approG}
&\gamma^{(r)}_{u,a}(\boldsymbol{\alpha}_r)=  \frac{\Tilde{\gamma}_u+(\alpha^{a}_{u,r})^2 K^2m\hat{\gamma}_u}{\sum_{u' \in \mathcal U_r, u' \neq u}\bigg(\Tilde{\gamma}_{u'}+(\alpha^{a}_{u',r})^2 K^2m\hat{\gamma}_{u'}\bigg)+1},
\end{align}
where $\Tilde{\gamma}_u=p\Tilde{G}_{u,a}/ \sigma^2_{\zeta_a}$ and $\hat{\gamma}_u=p\Gamma_u/\sigma^2_{\zeta_a}$.
Moreover, if UAV$_a$ does not have any direct links to the UEs in $\mathcal U_r$, \eref{approG} becomes 
\begin{align}\label{approNO}
&\gamma^{(r)}_{u,a}(\boldsymbol{\alpha}_r)=  \frac{(\alpha^{a}_{u,r})^2 K^2m\hat{\gamma}_u}{\sum_{u' \in \mathcal U_r, u' \neq u}\bigg((\alpha^{a}_{u',r})^2 K^2m\hat{\gamma}_{u'}\bigg)+1}.
\end{align}

\section{Problem Modeling}\label{PF} 
This section describes the original graph UAV network and the modified one after RIS deployment. Then, it formulates the proposed problem and discusses the proposed solution methodologies.

\subsection{Graph Modeling} 
\textbf{1) Original Graph:} The envisioned RIS-assisted NOMA UAV system is modeled using the undirected and weighted graph $\mathcal G_\text{org}(\mathcal V, \mathcal E_\text{org})$, where $\mathcal V$ represents the set of vertices associated with the network nodes (UAVs and UEs) and $\mathcal E_\text{org}$ represents the edges (links). We use the terms nodes, UAVs/UEs, and vertices interchangeably, and also the terms links and edges are used interchangeably.

For a graph edge $e_{l}$, $ 1 \leq l \leq E_\text{org}$, that connects two vertices $(v, v') \in \mathcal V$, we have
\begin{equation}
e_l = \begin{cases}
1 & \text{if}  ~  \gamma^\text{(UAV)}_{v,v'} \geq \gamma_\text{th}^\text{UAV} ~ \text{for}~ \text{UAV}_v \rightarrow \text{UAV}_{v'}~ \text{link},\\
1 & \text{if}  ~  \gamma^\text{(UA)}_{v,v'} \geq \gamma_\text{th}^\text{UE} ~ \text{for}~ \text{UE}_v \rightarrow \text{UAV}_{v'} ~ \text{link}, \\
0 & \text{otherwise},
\end{cases}
\end{equation}
where 
\begin{itemize}
    \item $\gamma^\text{(UAV)}_{v,v'}$ is the signal-to-noise ratio (SNR) for UAV-UAV communications. Basically, we use a LoS path component for each UAV-UAV link \cite{8292633, 9293155}, where the path-loss between UAV$_a$ and UAV$_{a'}$ can be expressed as
 $\text{PL}_{a,a'}=20\log\big( \frac{4 \pi f_c d_{a,a'}}{c} \big),$ where $d_{a,a'}$ is the distance between UAVs $a$ and $a'$, $f_c$ is the carrier frequency, and $c$ is the speed of light. The SNR in dB between UAV$_a$ and UAV$_{a'}$  is $\gamma^\text{(UAV)}_{a,a'}=10\log P-\text{PL}_{a,a'}-10 \log N_0$, where $N_0$ is the AWGN variance.

 \item   $\gamma^\text{(UA)}_{v,v'}$ is the SNR for UE-UAV communications. For $\text{UE}_u \rightarrow \text{UAV}_a$ channel,  
$\gamma^\text{(UA)}_{u,a}=\frac{p|\sqrt{\beta^\text{UA}_{u,a}}h^\text{UA}_{u,a}|^2}{N_0}$.

\item $\gamma^\text{UAV}_\text{th}$ and $\gamma^\text{UE}_\text{th}$ are the minimum SNR thresholds for the UAV-UAV and the UE-UAV communication links, respectively.
\end{itemize}
Let $\mathbf w \in {[\mathbb R^+]}^{E_\text{org}}$  be the weight vector of the UE-UAV and UAV-UAV links, which is defined as $\mathbf w= [ \omega_1, \omega_2, \ldots, \omega_{E_\text{org}}]$. Accordingly, for $e_{l}$, the weight $\omega_l$ is defined as 
\begin{equation}\label{wei}
\omega_l = \begin{cases}
\gamma^\text{(UAV)}_{v,v'} & \text{for}  ~  \text{UAV}_v \rightarrow \text{UAV}_{v'}~ \text{link},\\
 \gamma^\text{(UA)}_{v,v'} &  \text{for}~ \text{UE}_v \rightarrow \text{UAV}_{v'} ~ \text{link}.
\end{cases}
\end{equation}  

Let $\mathbf a_l$ be a vector of link $e_{l}$, where the $v$-th and the $v'$-th elements in $\mathbf a_l$ are given by $a_{v,l}=1$ and $a_{v',l}=-1$, respectively, and zero otherwise. Let $\mathbf A$ be the incidence matrix of $\mathcal G_\text{org}$ with the $l$-th column given by $\mathbf a_l$. Hence, in $\mathcal G_\text{org}(\mathcal V, \mathcal E_\text{org})$, the Laplacian matrix $\mathbf L_\text{org}$ is a $V \times V$ matrix, defined as \cite{4657335}
\begin{equation} \label{lap}
\mathbf L_\text{org}= \mathbf A  ~diag(\mathbf w) ~\mathbf A^T=\sum^{E_\text{org}}_{l=1} \omega_l \mathbf a_l \mathbf a^T_l,
\end{equation}
where the entries of $\mathbf L_\text{org}$ are given by
\begin{equation}
L_\text{org}(v,v') = \begin{cases}
\sum_{s \neq v}\omega_{v,s} &\text{if} ~v=v',\\
-\omega_{l} &\text{if}~ (v, v') \in \mathcal E_\text{org}, \\
0 & \text{otherwise}.
\end{cases}
\end{equation}
 Similar to \cite{new, 4657335, 4786516, 8292633, saifglobecom}, we choose a well-known metric, known as the \textit{algebraic connectivity}, also called the Fiedler value, denoted by $\lambda_2(\mathbf L_\text{org})$, to maximize the connectivity of the  network. 

\textbf{2) Modified Graph:}
With RIS deployment and partitioning, we define a new graph, denoted by $\mathcal G_\text{mod}(\mathcal V, \mathcal E_\text{mod})$, which has a larger set of edges denoted by $\mathcal E_\text{mod}$ with $\mathcal E_\text{mod}=\mathcal E_\text{org} \cup \mathcal E_\text{new}$, where $\mathcal E_\text{new}$ is the new  
$\text{UE}_u \xrightarrow{\text{RIS}_r} \text{UAV}_{a}$ edges, $\forall u\in \mathcal U_t, a \in \mathcal A, r\in \mathcal R$. For a new edge $e_{l}$, $ 1 \leq l \leq E_\text{mod}$, that connects two vertices $(u, a') \in \mathcal V$ via RIS$_r$, we have
\begin{equation}
e_l = \begin{cases}
1 & \text{if}  ~  \gamma^{(r)}_{u,a} \geq \gamma_\text{th}^\text{RIS} ~ \text{for}~ \text{UE}_u \xrightarrow{\text{RIS}_r} \text{UAV}_{a}~ \text{link},\\
0 & \text{otherwise},
\end{cases}
\end{equation}
where  $\gamma^\text{RIS}_\text{th}$ is the minimum SINR threshold of  $\text{UE}_u \rightarrow \text{UAV}_{a}$  via RIS$_r$. Essentially, with RIS deployment and partitioning, we either (i) add a new $\text{UE}_u \xrightarrow{\text{RIS}_r} \text{UAV}_{a}$ link if UE$_u$ is not directly connected to UAV$_a$ or (ii) tune the SINR of the $\text{UE}_u \rightarrow \text{UAV}_{a}$ link using the SINR of the cascaded link of RIS$_r$. As a result, for $e_{l}$, the weight $\omega_l$ is 
\begin{equation}\label{wei}
\omega_l = \begin{cases}
\gamma^{(r)}_{u,a} & \text{Given in \eref{approNO} or} \\
 \gamma^{(r)}_{u,a} &  \text{Given in \eref{approG}}.
\end{cases}
\end{equation}  
The network connectivity gain can be assessed by computing $\lambda_2 (\mathbf L_\text{mod}) \geq \lambda_2 (\mathbf L_\text{org})$, where $\mathbf L_\text{mod}$ is the resulting Laplacian matrix of the new graph $\mathcal G_\text{mod}(\mathcal V, \mathcal E_\text{mod})$. Adding more RIS-aided links can modify the entries of $\mathbf L_\text{org}$ and adds new non-zero elements in 
$\mathbf A$, which generally increases  $\lambda_2(\mathbf L_\text{mod})$. This is because the addition of RIS-aided links strengthens the overall connectivity of the network, making it more difficult for the resulting graph to be partitioned into disconnected subgraphs. %Therefore,   %The upper and lower bounds on the algebraic connectivity of a graph obtained by adding RIS-aided links connecting UEs to UAVs to a single connected graph is given in \cite[Proposition~3]{saifglobecom_E}.

\subsection{Node Reliability}\label{NRE}

%\ac{UAV-UAV links enhance UE-UAV connectivity by providing link redundancy and routing, enabling cooperation, and improving network flexibility. If UAV-UAV links are ignored, the model becomes simpler and it cannot represent   the original problem, as it can no longer capture the full resilience and connectivity benefits provided by UAV cooperation.}

Let $\mathcal G^{-v}_\text{org}$ be the remaining graph after removing vertex $v \in \mathcal V$ along with all its adjacent edges to other vertices in $\mathcal G_\text{org}$, i.e., $\mathcal G^{-v}_\text{org} \subseteq \mathcal G_\text{org}$.   We calculate the connectivity of the remaining graph based on the \textit{Fiedler value} \cite{new}, which is defined as $\lambda_2(\mathcal G^{-v}_\text{org})$. A node that, when removed, significantly reduces the connectivity of the network is declared to be highly critical and thus not reliable.   Therefore,  we measure the reliability of the nodes based on their criticalities, which reflects the  severity of the impact on the connectivity of the remaining graph, which is defined  as 
\begin{equation}\label{eq1}
C_v=\lambda_2(\mathcal G^{-v}_\text{org}).
\end{equation}
Equation \eref{eq1} implies that highly critical nodes are not reliable. If $C_v > C_{v'}$, node $v$ has higher reliability than node $v'$. Since UAVs represent the backhaul core of the considered uplink system, they are the most critical nodes that can severely impact the network connectivity if they fail. Thus, we consider the reliability of the UAVs only. By utilizing RISs, signals of UEs can be redirected to reliable UAVs, resulting in more resilient and connected network. %In Fig.~\ref{fig2}, we consider a single unweighted graph component with $6$ nodes, where we remove node $3$ in Fig.~\ref{fig2}(a)  and remove node $5$ in Fig.~\ref{fig2}(b). The dash lines represent the removed edges after removing the corresponding node, and accordingly we calculate the reliability of the node using \eref{eq1}, resulting in $C_3=0.38$ and   $C_5=1.59$. This highlights that node $5$ is less critical than node $3$, which is more reliable for connecting blocked UEs via  RIS.   

% \begin{figure}[t!]  
% \begin{center}
% \includegraphics[width=0.95\linewidth, draft=false]{Example_C.pdf}
% \caption{An example of node reliability with a single unweighted graph component.}
%    \label{fig2}
% \end{center}
 
% \end{figure}

\subsection{Problem Formulation}
Let $\mathbf X$ be the binary RIS-UAV  assignment matrix with entries $x^{(r)}_{a}$, where  $x^{(r)}_{a}$ is $1$ if UAV$_k$ is connected to RIS$_r$, and  $x^{(r)}_{a}=0$ otherwise. Likewise, let $\mathbf Z$ be the  UE-RIS assignment matrix and comprise binary elements $z^{(r)}_u$  defined as $1$ if  $\text{UE}_u$ is connected to RIS$_r$, and $0$ otherwise.  The proposed optimization problem is formulated as  
\begin{subequations} \nonumber 
\label{eq10}
\begin{align}
&\mathcal{P:} ~\max_{\mathbf Z, \mathbf X, \boldsymbol{\alpha}} ~~~~ \lambda_2(\mathbf L_\text{mod} (\mathbf Z, \mathbf X, \boldsymbol{\alpha}))
\label{eq10a}\\
 &~~~~~~~~~~{\rm s.~t.\ }\\
 &\text{C$_1$:} ~~~~\sum_{r \in \mathcal R} z^{(r)}_{u} \leq 1, ~~~~~~~~~~~~~~ \forall u \in \{1, 2, \dots,  U_t\},\\
  &\text{C$_2$:} ~~~~~\sum_{u \in \mathcal U_t} z^{(r)}_u \leq U_r, ~~~~~~~~~~~\forall r \in \{1, 2, \dots,  R\},\\
  &\text{C$_3$:} ~~~~~\sum_{r \in \mathcal R} \sum_{u \in \mathcal U_t} z^{(r)}_u \leq U_rR,\\
  &\text{C$_4$:} ~~~~\sum_{r \in \mathcal R}   x^{(r)}_{a} \leq 1, ~~~~~~~~~~~~~~\forall a \in \{1, 2, \dots,  A\},\\
   &\text{C$_5$:} ~~~~\sum_{r \in \mathcal R} \sum_{a \in \mathcal A} x^{(r)}_{a} \leq R, \\
  &\text{C$_6$:} ~~~~~\sum_{a \in \mathcal A}  x^{(r)}_{a} z^{(r)}_u \leq 1, ~~~~~~~~~~~\forall r \in \mathcal R, \forall u \in \mathcal U_r,\\
 &\text{C$_7$:} ~~~~~\gamma^{(r)}_{u,a}(\mathbf Z, \mathbf X, \boldsymbol{\alpha}_r)\geq C_a\gamma^\text{RIS}_\text{th}, ~~~~~~~~~\forall (u,a,r),\\
&\text{C$_8$:}~~~~~\sum_{u=1}^{U_r}\alpha^{a}_{u,r} z^{(r)}_ux^{(r)}_{a} \leq 1, ~~~~~~\forall r \in \mathcal R, \forall u \in \mathcal U_r,  \\
&\text{C$_9$:}~~~~~  0 \preceq \boldsymbol{\alpha}_r \preceq 1, ~~~~~~~~~~~~~~\forall r \in \{1, 2, \dots,  R\},\\
& ~~~~~~~~~~   x^{(r)}_{k}, z^{(r)}_u \in \{0,1\}, ~\forall u \in \mathcal U_r, ~~r \in \mathcal R, a \in \mathcal A,
\end{align}
\end{subequations} 
where $\preceq$ is the pairwise inequality  and $\boldsymbol{\alpha}=[\boldsymbol{\alpha}_r, \ldots, \boldsymbol{\alpha}_R]$. In $\mathcal P$, \text{C$_1$} shows that the transmitting UE$_u$ is connected to only one RIS; \text{C$_2$} ensures that at most $U_r$ connections exist between the UEs and each RIS. Thus, \text{C$_1$} and \text{C$_2$} make \text{C$_3$} with at most $U_rR$ links are created between the transmitting UEs and the RISs. \text{C$_4$} implies that UAV$_a$ is connected to only one RIS, and accordingly,  \text{C$_5$} implies that at most $R$ reflected links should be created for the UAVs  via the RISs. \text{C$_6$}  means  that a complete path can be created from a selected UE to the suitable UAV through one RIS.
 \text{C$_7$} states that the SINR of a typical $\text{UE}_u \xrightarrow{\text{RIS}_r} \text{UAV}_{a}$ link should be greater than or equal to the minimum SINR threshold of that link times the reliability value of the corresponding UAV, constituting the QoS constraint on UAV$_a$ based on its reliability. Particularly,  the reliability metric $C_a$ controls the QoS limit set for UAV$_a$, $\forall a$.  \text{C$_8$} ensures that the allocated portion of RIS$_r$ does not exceed unity, i.e., the total number of allocated RIS$_r$ elements is not higher than the total number of RIS$_r$ elements. \text{C$_9$} outlines the domain of optimization variables, meaning that  $ 0 \leq \alpha^{a}_{u,r} \leq 1,$ $\forall (u,a,r)$.

The proposed optimization problem $\mathcal P$ is a mixed integer nonlinear programming (MINLP) problem due to the combinatorial nature of the optimization variables. Obtaining a globally optimal solution of a MINLP is generally difficult even for a moderate size of a network, which necessitates an efficient solution.

\subsection{Problem Methodology}\label{SOM}
If a segment of RIS$_r$ is allocated to aid and reflect the signal of UE$_u$ to UAV$_a$ (i.e., $z^{(r)}_ux^{(r)}_{a}=1$ and $0 \leq \alpha^{a}_{u,r} \leq 1$), a new RIS-aided link $l$ (i.e., $\text{UE}_u \xrightarrow{\text{RIS}_r} \text{UAV}_{a}$) is added to the graph network, we will have
\begin{equation} \label{L'}
\mathbf L_\text{mod}(\boldsymbol{\alpha}_r, z^{(r)}_u, x^{(r)}_{a})=\mathbf L_\text{org}+ z^{(r)}_ux^{(r)}_{a} \omega^o_{u,a} \mathbf a_{u,a} \mathbf a^T_{u,a},
\end{equation}
where $\mathbf a_{u,a}$ is the incidence vector resulting from adding an UE-UAV link to $\mathcal G$ and $\omega^o_{u,a}$ is the weight of a constructed UE-RIS-UAV link, given as
\begin{equation}\label{WE}
\omega^o_{u,a}=\gamma^{(r)}_{u,a}(\boldsymbol{\alpha}_r).
\end{equation}
%Here, link $l$ is added if $z^{(r)}_u x^{(r)}_{u}=1$ and $0 \leq \alpha^{a}_{u,r} \leq 1$ for $u=\left \lceil \frac{l}{|\mathcal R| \times |\mathcal A|}\right \rceil$, $r=\left \lceil \frac{l-(u-1)|\mathcal R| \times |\mathcal A|}{|\mathcal A|}\right \rceil$, and $a=l-(u-1)|\mathcal R| \times |\mathcal A|-(r-1)|\mathcal A|$. As it can be seen, $u$, $r$, and $a$ can uniquely be determined for a given value of $l$. For simplicity, we use the term $\omega^o_l$ instead of  $\gamma^{(r)}_{u,a}(\boldsymbol{\alpha}_r)$ unless when it creates confusion. 

In this paper, we are interested in maximizing network connectivity by solving $\mathcal P$ under the max-sum link quality, where we maximize the  weights of the added $\text{UE}_u \xrightarrow{\text{RIS}_r} \text{UAV}_{a}$ links. By using \eref{L'}, we get an equivalent objective function of $\mathcal P$ as follows $\lambda_2(\mathbf L_\text{mod} (\boldsymbol{\alpha}, \mathbf Z, \mathbf X))=  \lambda_2( \mathbf L_\text{org}+ \sum_{\substack{r\in \mathcal R}} \sum_{\substack{a\in \mathcal A}}\sum_{u\in \mathcal U_r}z^{(r)}_ux^{(r)}_{a} \omega^o_{u,a}(\boldsymbol{\alpha}) \mathbf a_{u,a} \mathbf a^T_{u,a})$. Since $\mathbf L_\text{org}$ of the original network graph is fixed and network connectivity is a monotonically
increasing function of the added links and their weights \cite{new}, \cite{saifglobecom_E}, we can  maximize
the network connectivity by selecting the $\text{UE}_u \xrightarrow{\text{RIS}_r} \text{UAV}_{a}$ links and maximizing their sum SINR.

Subsequently, $\mathcal P$, that aims at the joint optimization of RIS-aided link selection (i.e., UE-RIS-UAV clustering) and RIS partitioning  to maximize the connectivity via maximizing the sum SINR of the new $\text{UE}_u \xrightarrow{\text{RIS}_r} \text{UAV}_{a}$ links, can be given by 
\begin{subequations} \nonumber 
\label{eq10}
\begin{align}
& \mathcal P^{\text{sum}}: \max_{\boldsymbol{\alpha}, \mathbf Z, \mathbf X} ~~ \lambda_2\bigg(\mathbf L_\text{org} + \sum_{\substack{r\in \mathcal R \\ u\in \mathcal U_r}} \sum_{\substack{a\in \mathcal A}} z^{(r)}_ux^{(r)}_{a} \omega^o_{u,a}(\boldsymbol{\alpha}) \mathbf a_{u,a} \mathbf a^T_{u,a}\bigg)
\label{eq10a}\\
 &~~~~~~~{\rm s.~t.\ }
 ~~~~~ \text{C$_1$} - \text{C$_9$}, 
\end{align}
\end{subequations} 
which uses all constraints of $\mathcal P$.

In this paper, we propose a novel optimization technique  to maximize network connectivity.  The solution methodology for solving $\mathcal P^{\text{sum}}$ effectively  can be   decomposed into the following two subproblems:
\begin{itemize}
    \item Given the initial RIS partitioning $\boldsymbol{\alpha^0}$, $\mathcal P^\text{sum}$ can be solved with the UE NOMA and UAV clustering variables $\mathbf Z$ and $\mathbf X$, which is still MILNP. Thus, we propose an  innovative procedure of two types of clustering: RIS-UAV clustering, denoted as $\mathcal C$ and UE-$\mathcal C$ clustering. The clustering  procedure is presented in \sref{Link}, and its pseudo code  is provided in Algorithm \ref{Algorithm1}.
     
    \item Given the formed clustering $\mathbf Z$ and $\mathbf X$, we focus on RIS partitioning $\boldsymbol{\alpha}$, which is derived in a closed-form in Proposition $1$.  A detailed step-by-step closed-form solution for the RIS partitioning is presented in \sref{RIS}.

\end{itemize}
The proposed iterative solution have two procedures, summarized in Algorithm \ref{alg_all}, that run  as follows: (i) RIS-aided link selection and (ii) RIS partitioning design. These two procedures are executed seamlessly until convergence in a centralized manner at the BS. The BS collects the CSI of the system to perform the RIS-aided link selection and adjusts the phase shifts of the RIS elements. Then, it optimizes the RIS partitions. In the next time slot, these two procedures are re-executed for different network topologies.

\begin{algorithm}[t!]
	\caption{Iterative Approach: Finding $\boldsymbol{\alpha}$, $\mathbf Z$, and $\mathbf X$}
	\label{alg_all}
	\begin{algorithmic}[1]
		\State  Initialize $\boldsymbol{\alpha^t}$ and $t=0$;  
      \Repeat
       
       \State Update $\mathbf Z$ and $\mathbf X$ as in Algorithm \ref{Algorithm1} and \eref{BB};
       \State Optimize $\boldsymbol{\alpha}$  as  in  \pref{pro1}

    \State Calculate $\lambda_{2,t}(\mathbf L_\text{mod}(\mathbf Z, \mathbf X, \boldsymbol{\alpha}))$;
    \Until{$|\lambda_{2,t}(\mathbf L_\text{mod}(\mathbf Z, \mathbf X, \boldsymbol{\alpha}))-\lambda_{2,t-1}(\mathbf L_\text{mod}(\mathbf Z, \mathbf X, \boldsymbol{\alpha}))|< \delta$}
	\end{algorithmic}
\end{algorithm}

\section{RIS-aided Link Selection: UE-RIS-UAV Clustering}\label{Link}
Given RIS partitioning $\boldsymbol{\alpha^0}$, the subproblem for updating $\mathbf Z$ and $\mathbf X$ is given by 
\begin{subequations} \nonumber 
\label{eq10}
\begin{align}
&\mathcal P^\text{sum}_1: \max_{\mathbf Z, \mathbf X} ~~ \lambda_2\bigg(\mathbf L_\text{org} + \sum_{\substack{r\in \mathcal R}} \sum_{\substack{u \in \mathcal U_r}} \sum_{a\in \mathcal A} z_u^{(r)} x_a^{(r)} \omega^o_{u,a} \mathbf a_{u,a} \mathbf a^T_{u,a}\bigg)
\label{eq10a}\\
 &~~~~~~~{\rm s.~t.\ }
 ~~~~~ \text{C$_1$} - \text{C$_7$}. 
\end{align}
\end{subequations}  
This section solves $\mathcal P^\text{sum}_1$ to optimize $\mathbf Z$ and $\mathbf X$ through an innovative clustering solution. Specifically, we propose two procedures of clustering: RIS-UAV clustering and NOMA UE clustering with size $U_r$, which are presented in Algorithm $\ref{Algorithm1}$ and presented as follows.

%\subsection{UE NOMA Clustering} \label{Clustering}
\textbf{UAV-RIS clustering:} First, we assign the most reliable UAVs to the RISs $\mathcal R$, where each RIS cluster has at most one reliable UAV. We sort the UAVs in an ascending order based on their reliability $C_a, \forall a \in \mathcal A$. This sorted UAV index set is denoted as $\Tilde{\mathcal A}$. Afterwards, the first $R$ UAVs in the set $\Tilde{\mathcal A}$  are associated with the closest RISs.  The formed UAV-RIS clusters are denoted as $\mathcal {C}=\{\mathcal C_1, \ldots, \mathcal C_R\}$, where each cluster refers to one RIS and one reliable UAV, i.e., $x^{(r)}_a=1, \forall (r,a)$. For simplicity, each cluster is indexed by RIS$_r, \forall r$, and we omit the index $a$ hereafter unless when it creates confusion.

\textbf{UE NOMA clustering with size $U_r$:}  We cluster the UEs to the RIS clusters (NOMA clusters) each with size $U_r$, where $U_r$ possible transmitting UEs are associated with one cluster. First, we sort the UEs  in an ascending order based on the RSS of their direct links to the UAVs. This sorted UE index set is denoted as $\Tilde{\mathcal U}$. Such sorted set is crucial to maximize network connectivity and ensure that UEs with weak connections remain connected, thereby avoiding zero connectivity. Specifically, UEs with weak or no direct links to UAVs are prioritized to be clustered to the RISs. As such, we provide alternative transmission paths for UEs with weak channels to UAVs via RISs, enhancing their ability to transmit signals effectively. This sorted UE index set plays a crucial role, where it is
employed to partition the set of the first $RU_r$ weakly connected UEs into $U_r$ subsets labeled as $\mathtt U_i=\Tilde{\mathcal U}\{(i-1)R+1: iR\}, i \in [1, U_r]$. In Line $16$, we initialize the clusters. Specifically, $\mathcal C_r$ is initialized
using the $r$-th element extracted from the first subset of $\mathtt U_1$, $\forall r\in \{1, \ldots, R\}$, i.e., $\mathcal C_r=\mathtt U_1\{ r\}$.

Following, the UEs are iteratively admitted to NOMA clusters as illustrated in the subsequent process:
(1) Line $18$ updates the cluster by temporarily admitting the waited UEs and (2) Line $19$ calls  the utility matrix procedure to compute the utility matrix $\mathbf O_i$ for the $i$-th round of UE admission. The utility matrix  function $\mathbf O_i$ is calculated using nested ``for loops'' between Lines $27$ and $33$. In each iteration of the
inner loop (indexed by $u$), Line $29$ creates a temporary cluster $\mathcal Q_r$ by admitting the UE$_u$ of the set $\mathcal I_i$ into $\mathcal C_r$. (4) Lines $30$ and $31$ compute the values of $O_u^{(r)}$ based on the objective
function of interest by evaluating the sum SINRs of the temporary cluster $\mathcal Q_r$, respectively. (5) Lines $20$, $21$, and $22$ utilize the procedures to identify the final
admissions of the UEs that maximize the sum SINR of all clusters, respectively. These two final assignment procedures ensure that the same UE is assigned to only one NOMA cluster, i.e., the UEs are assigned to the appropriate NOMA clusters. For such assignments, we use the following LSA  optimization of the $i$-th round (i.e., $\mathcal I_i$)
\begin{subequations}   \label{BB}
\begin{align}
& \mathbf B^*_\text{sum} \leftarrow ~\max_{\mathbf B}   \sum_{r\in \mathcal R}\sum_{u\in \mathcal I_i}  b_u^{(r)} O_u^{(r)} 
\label{eq10a}\\ \nonumber 
 &~~~~~~~{\rm s.~t.\ } \sum_{r\in \mathcal R}b_i^{(r)}=1, \forall u\in \mathcal I_i,\\ & \nonumber 
 ~~~~~~~~~~~~~\sum_{u\in \mathcal I_i}b_u^{(r)} =1, \forall r\in \mathcal R. \nonumber 
\end{align}
\end{subequations} 

\textbf{Computational complexity:} The complexity of the proposed UE NOMA clustering
is dominated by the complexity of the utility matrix and the LSA procedures. The RIS assignment and utility matrix calculation complexity for NOMA clustering  is given by $\mathcal O(R)$ and $\mathcal O(\sum_{i=1}^{U_r-1}RI_i)$, respectively. Considering the well-known cubic
complexity of the LSA solution \cite{comp}, the overall complexity
can be obtained as $\mathcal O(\sum_{i=1}^{U_r-1}RI_i+R+(\max(R, I_i))^3)$, which constitutes the entire complexity
of Algorithm \ref{Algorithm1}.

 \begin{algorithm}[t!]
	\caption{UE NOMA Clustering}
	\label{Algorithm1}
	\begin{algorithmic}[1]
		\State \textbf{Input:} $\mathcal U$, $\mathcal A$, $\mathcal R$, $U_r, \forall r$
   %\algrule
       \State \textbf{Step 1: UAV-RIS Clustering $(\mathcal R, \mathcal A)$}
       \State Compute the distances between RISs and UAVs
       \State $\mathcal A_0 \leftarrow{}~ \{a\in \mathcal A ~|~  d_{r,a} \leq R^\text{RA}_\text{RIS}, \forall r\}$
       \State $\Tilde{\mathcal A} \leftarrow{}~ \text{SortAscend}(C_a, \forall a \in \mathcal A_0)$ //\textit{UAV ordering based on their reliability}
       \State  $\Tilde{\mathcal A}_0= \Tilde{\mathcal A}\{1:R\}$
       \For{$r=1: R$}
        
       \State $\mathcal C_r \leftarrow{} \bigg(\Tilde{\mathcal A}_0(r), \text{Closest RIS}\bigg)$
       \State $\mathcal C \leftarrow \mathcal C \cup \mathcal C_r$
       \EndFor
      %\algrule
      \State \textbf{Step 2: UE NOMA Clustering  $(\mathcal C, \mathcal U)$}
       \State Compute the distances between UEs and RISs
       \State $\mathcal U_0 \leftarrow{}~ \{u\in \mathcal U ~|~  d_{u,r} \leq R^\text{UR}_\text{RIS}, \forall r\}$
       \State $\Tilde{\mathcal U} \leftarrow{}~ \text{SortDescend}(\text{RSS}, \forall u \in \mathcal U)$ //\textit{UE ordering based on their RSS to the UAVs}

      \State $\mathtt U_i \leftarrow \Tilde{\mathcal U}\{(i-1)R+1: iR\}, i \in [1, U_r]$ //Form UE NOMA clusters

      \State $\mathcal C_r \leftarrow \mathtt U_1\{r\}$, $\forall r=\{1, \ldots, R\}$ //Initialize UE NOMA clusters

       \For{$i=1:U_r-1$}

        \State $\mathcal I_i \leftarrow \mathtt U_{i+1}$

        \State $\mathbf O_i \leftarrow \text{Utility Matrix UE}(\mathcal C_r, \mathcal I_i)$ //Create utility matrix for UE NOMA
clustering

       \State  $\mathbf B^i_\text{sum} \leftarrow \mathbf O_i$
    
     \State $\mathcal C_r \leftarrow \mathcal C_r \cup \mathcal I_i\{u\}$, $y_u^{(r)}=1, \forall (r,u)$~//Update clusters

     \State $z^{(r)}_{u} \leftarrow 1$, $\forall u \in \mathcal C_r$
     \EndFor 
     \State \textbf{return} $\mathbf Z$, $\mathbf X$, and $\mathcal C_r, \forall r$
\algrule
     \State $\textbf{Utility Matrix UE}(\mathcal C_r, \mathcal I_i)$
   
     \State $\mathbf O_i \leftarrow \boldsymbol{0}^{R\times I_i}$

    \For{$r=1:r$}
    \For{$u=1:I_i$}
     \State $\mathcal Q_r \leftarrow \mathcal C_r \cup \mathcal I_i\{u\}$   //\textit{Temp. admit u-th UE of $\mathcal I_i$ to $\mathcal C_r$}

    \State $\gamma^{(r)}_{j} \leftarrow$ \text{Obtain SINR for UE$_j$ from RIS$_r$} using   $\boldsymbol{\alpha^0_r}$, $\forall j \in \mathcal Q_r$
 
       \State $O^{(r)}_u \leftarrow \sum_{j\in \mathcal Q_r} \gamma^{(r)}_u$
     
    \EndFor
     \EndFor
   \State \textbf{return} $\mathbf O_u$  
      %\algrule
     % \State \textbf{procedure RIS Partitioning  $(\mathcal M, K_m)$} $\rightarrow$ Algorithm \ref{Algorithm2}
    
     % \algrule
     % \State \textbf{procedure RIS Deployment} $\rightarrow$ Algorithm \ref{Algorithm3}
      
	\end{algorithmic}
\end{algorithm}

\section{Optimal RIS Partitioning}\label{RIS}
%In NOMA schemes, UAVs utilize SIC to decode signals transmitted by UEs in descending order of received signal power. Traditionally, the grant-based NOMA scheme ensures that the UE with the highest channel gain is assigned the highest transmit power, enabling other cluster members to create interference and improve their SINRs.
%This section proposes two novel transmission schemes at the UEs' side while fine-tuning the received power at the UAVs through the optimization of RIS partitions.

%\subsection{RIS Partitioning}
We maximize the sum SINR of all the RIS-aided links within NOMA cluster under the given SINR constraint and RIS elements allocation. Therefore, given $\mathbf Z$ and $\mathbf X$, the problem of optimal RIS partitioning to maximize the sum SINR of the RIS-aided links within the NOMA cluster can
be formulated as
\begin{subequations} \nonumber 
\label{sum-rate}
\begin{align}
&\mathcal P^\text{sum}_2: ~\max_{\boldsymbol{\alpha_r}} ~~~~  \sum_{u=1}^{U_r}\gamma^{(r)}_{u,a}(\alpha^{a}_{u,r})  
\label{eq10a}\\
 &~~~~~~~{\rm s.~t.\ }\\
 &\text{C$_1^2$:} ~~~~~\gamma^{(r)}_{u,a}(\alpha^{a}_{u,r})\geq C_a\gamma^\text{RIS}_\text{th}, ~~~~~~~u \in \{1, \ldots, U_r\},\\
 & \text{C$_2^2$:} ~~~~~ 0 \leq \alpha^{a}_{u,r} \leq 1, ~~~~~~~~~~~~~~~~u \in \{1,\ldots, U_r\},\\
&\text{C$_3^2$:}~~~~~\sum_{u=1}^{U_r}\alpha^{a}_{u,r} \leq 1. ~~~~~~~   
\end{align}
\end{subequations}

The closed form solution $\alpha^{a^*}_{u,r}, u \in \{1, \ldots, U_r\}$ of the partitions of RIS$_r$ is provided in the following Proposition.
\begin{proposition} \label{pro1}
The optimal RIS$_r$ partitioning that provides the maximized network connectivity via maximizing the sum SINR for the RIS$_r$ cluster is given by
\begin{align}\label{closed}
\alpha^{a*}_{u,r} = \begin{cases}
\sqrt{\left [\frac{C_a \gamma^\text{RIS}_\text{th}\sum\limits_{\substack{u\in \mathcal U_r \\ u'\neq u}}\bigg(\Tilde{\gamma}_{u'}+(\alpha^{a^*}_{u',r})^2 B+1\bigg)-\Tilde{\gamma}_u}{A} \right]^+} & \text{if}~   u \neq 1,\\
1-\sum_{u'=u+1}^{U_r}\alpha^{a^*}_{u',r} & \text{if}~    u=1,
\end{cases}
\end{align}
where $A=K^2m\hat{\gamma}_{u}$, $B=K^2m\hat{\gamma}_{u'}$, and $[x]^+=\max\{x,0\}$.
 \end{proposition}
\begin{proof}
We maximize the sum SINR of the set RIS-aided links  $\text{UE}_u \xrightarrow{\text{RIS}_r} \text{UAV}_{a}$ of RIS$_r$
given  \text{C$_2^2$} and \text{C$_3^2$} 
within the constraint of the minimum SINR requirement of each UE. Our allocation scheme is to allocate RIS elements
to all UEs except for UE with the highest direct UE-UAV link quality. The SINR of the RIS-aided links of those UEs are just equal to their
minimum required SINR. All the remaining RIS elements will be allocated to the UE with the highest direct UE-UAV link quality. %In fact, the UE with the maximum channel gain contributes the most to the sum-rate of a NOMA cluster; thus, it should be allocated the maximum number of elements. 

For simplicity of the  analysis, we explain the steps of the closed-form solution through an example of $4$ UEs (i.e., $U_r=4$) with a reverse chronological order, e.g., \cite{NOMA-1, optimal}.  Suppose that  $\Tilde{\gamma}_1> \Tilde{\gamma}_2>\Tilde{\gamma}_3 > \Tilde{\gamma}_4$, i.e., UE$_1$ is the UE with the highest channel gain, where its SINR is given in \eqref{approG}. \\
\textbf{UE$_4$:} If $\Tilde{\gamma}_4 \geq  C_a \gamma^\text{RIS}_\text{th}$, $\alpha^{a^*}_{4,r} =0$, otherwise $\Tilde{\gamma}_4+(\alpha^{a}_{4,r})^2 K^2m\hat{\gamma}_4=C_a \gamma^\text{RIS}_\text{th}$. By performing some manipulations, we have $\alpha^{a*}_{4,r}= \sqrt{\frac{C_a \gamma^\text{RIS}_\text{th}-\Tilde{\gamma}_4}{K^2m\hat{\gamma}_4}}$.\\
\textbf{UE$_3$:} If $\Tilde{\gamma}_3 \geq  C_a \gamma^\text{RIS}_\text{th}$, $\alpha^{a^*}_{3,r} =0$, otherwise $\frac{\Tilde{\gamma}_3+(\alpha^{a}_{3,r})^2 K^2m\hat{\gamma}_3}{\Tilde{\gamma}_{4}+(\alpha^{a^*}_{4,r})^2 K^2m\hat{\gamma}_{4}+1}=C_a \gamma^\text{RIS}_\text{th}$. By performing some manipulations, we have $\alpha^{a*}_{3,r}= \sqrt{\frac{C_a \gamma^\text{RIS}_\text{th}\bigg[ \Tilde{\gamma}_{4}+(\alpha^{a^*}_{4,r})^2 K^2m\hat{\gamma}_{4}+1\bigg]-\Tilde{\gamma}_3}{K^2m\hat{\gamma}_3}}$.
\\
\textbf{UE$_2$:} If $\Tilde{\gamma}_2 \geq  C_a \gamma^\text{RIS}_\text{th}$, $\alpha^{a^*}_{2,r} =0$, otherwise $\frac{\Tilde{\gamma}_2+(\alpha^{a}_{2,r})^2 K^2m\hat{\gamma}_2}{\Tilde{\gamma}_{3}+(\alpha^{a^*}_{3,r})^2 K^2m\hat{\gamma}_{3}+\Tilde{\gamma}_{4}+(\alpha^{a^*}_{4,r})^2 K^2m\hat{\gamma}_{4}+1}=C_a \gamma^\text{RIS}_\text{th}$. By performing some manipulations, we have $\alpha^{a*}_{2,r}= \sqrt{\frac{C_a \gamma^\text{RIS}_\text{th}\sum_{u=3,4}\bigg[ \Tilde{\gamma}_{u}+(\alpha^{a^*}_{u,r})^2 K^2m\hat{\gamma}_{u}+1\bigg]-\Tilde{\gamma}_2}{K^2m\hat{\gamma}_2}}$.\\
\textbf{UE$_1$:} If $\Tilde{\gamma}_1 \geq  C_a \gamma^\text{RIS}_\text{th}$, $\alpha^{a^*}_{1,r} =0$, otherwise $\alpha^{a^*}_{1,r}=1-\alpha^{a^*}_{2,r}-\alpha^{a^*}_{3,r}-\alpha^{a^*}_{4,r}$.

In general, we have $\alpha^{a*}_{u,r}= \sqrt{\left [\frac{C_a \gamma^\text{RIS}_\text{th}\sum_{u'=u+1}^{U_r}\bigg[\Tilde{\gamma}_{u'}+(\alpha^{a^*}_{u',r})^2 K^2m\hat{\gamma}_{u'}+1\bigg]-\Tilde{\gamma}_u}{K^2m\hat{\gamma}_u} \right]^+}, \forall u \neq 1$. In addition, for the highest channel gain UE$_1$, $\alpha^{a*}_{1,r}= 1-\sum_{u'=u+1}^{U_r}\alpha^{a*}_{u',r}$, which completes the proof. The optimal solutions are subject to the constraint $\sum_{u=1}^{U_r}\alpha^{a^*}_{u,r} = 1$. 
\end{proof}

%\section{Overall Algorithm and Computational Complexity}\label{pert}

\section{Numerical Results}\label{NR}
 
This section  presents simulation results to assess the
effectiveness of the proposed RIS-aided UAV NOMA system. Monte Carlo (MC) simulations are employed, where each presented outcome is an average of over $500$ random network realizations. Following \cite{44A, 45A}, we compute all large-scale path loss values using the 3GPP Urban Micro (UMi) model\cite{Urban} at a frequency of $3$ GHz. Additionally, we assume a Nakagami$-f$ shape parameter of
$f_1=f_2=5$ for the cascaded UE-RIS-UAV link and $f_u=1$ for the direct UE-to-UAV links, with a spread parameter of unity for all the links. Similar to \cite{saifglobecom}, we use $\sqrt{\frac{\beta_0}{(d_{u,r})^2}}$ and $\sqrt{\frac{\beta_0}{(d_{r,a})^2}}$ for $\text{UE}_u \rightarrow \text{RIS}_r$ and $\text{RIS}_r \rightarrow \text{UAV}_{a}$ links, respectively, where $\beta_0$ denotes the path loss at the reference distance $d_\text{ref}=1$ m  and $d$ is the corresponding distance. The channel bandwidth is termed as $B$. The successful communications between the UEs and RISs and between the RISs and UAVs  are measured using the distance threshold $R^\text{UR}_\text{RIS}$ and $R^\text{RA}_\text{RIS}$, respectively, i.e., UE$_u$ can be connected to RIS$_r$ with distance $d_{u,r}$ if $d_{u,r} \leq R^\text{UR}_\text{RIS}$ and RIS$_r$ can be connected to UAV$_a$ with distance $d_{r,a}$ if $d_{r,a} \leq R^\text{RA}_\text{RIS}$. The RISs  and UAVs are located at altitudes of $120, 200$ m, respectively. Unless stated otherwise, we use the default system parameters as listed in Table \ref{table_1}.

\begin{table}[t!] 
	\renewcommand{\arraystretch}{0.9}
	\caption{Simulation Parameters}
	\label{table_1}
	\centering
	\begin{tabular}{||p{1.5cm}| p{1.8cm}|p{1.5cm}| p{1.8cm}||}
		\hline
		
		\textbf{Parameter} & \textbf{Value} & \textbf{Parameter} & \textbf{Value}\\
         \hline 
         \hline
		  $c$ & $3\times 10^8$ m/s & $f_c$ & $3$ GHz  \\
		\hline
		$P$ & $1$ watt   & $p$ & $23$ dBm \\   
       \hline
       %$\beta_1$ & $0.9$ & $\beta_2$ & $0.999$ \\
       \hline 
       %$\epsilon$  & $10^{-8}$ & $\nu$ & $0.001$ \\ 
       %\hline 
       $\beta_0$ & $10^{-2}$ & $\gamma^\text{RIS}_\text{th}$ & $30$ dB \\
       %\hline
       %$\gamma_0^\text{UE}$ & $83$ dB & $\gamma_0^\text{UAV}$ & $80$ dB \\
       \hline
       $B$ & $250$ KHz   & $\sigma^2_{\zeta_k}$ & $-130$ dBm \\
      \hline
       $R^\text{RA}_\text{RIS},  R^\text{UR}_\text{RIS}$ & $100, 150$ m & $N_0$ & $-130$ dBm\\
       \hline
	\end{tabular}
	
\end{table}

%\vspace{-0.4cm}

\subsection{Exact and Approximate SINR Models}
In Fig.~\ref{fig3}, we show the sum rate values obtained by comparing the exact SINR and the SINR formulated by the
approximation, as represented by \eref{exact} and \eref{approG}, respectively, for both  perfect (Fig.~\ref{fig3}(a)) and imperfect (Fig.~\ref{fig3}(b)) RIS phase shifts.  In this part, we consider $4$ UEs, $2$ UAVs, and $2$ RISs, where we set the 3D Cartesian coordinates (in meters) for the UEs as $(318, 200, 0)$, $(100, 50, 0)$, $(150, 170, 0)$, and $(200, 220, 0)$, while  the two UAVs are located at $(460, 340, 200)$ and $(370, 14, 200)$, respectively. The two RISs are placed at $(0, 0, 120)$ and $(100, 100, 120)$, respectively. We consider RIS-aided links, and the RIS partitioning is calculated based on Proposition $1$.

The results of Fig.~\ref{fig3}(a) indicate that the approximate  rates match the exact results, while the approximate rate closely aligns with the exact rate for the imperfect phase shifter of Fig.~\ref{fig3}(b), demonstrating the precision of the approximate model used for the analytical formulation in Section II. Specifically, there is a slight reduction in  the approximate
rate compared to the exact rates in the case of imperfect phase shifter of Fig. \ref{fig3}(b), noting that the term $\sqrt{\Tilde{G}_{u,r} \Tilde{H}_{r,a}}  \sum_{u' \in \mathcal U_r, u' \neq u}\bigg( \sum_{k'=1}^{K^{a}_{u',r}}h^{(k')}_{u,r,a}e^{j\theta^{(k')}_{u',r}}\bigg)$ reduces the exact rate compared to the approximated rate.  Such an approximate rate can be served as a lower bound to the exact rate. Fig.~\ref{fig3} further allows us to rely on the approximate analytical results to study network connectivity. As illustrated in Fig.~\ref{fig3}, the first-order reflections from the non-aligned portions of the allocated RIS and from other RISs in the network have negligible impact on the received signal. This implies that the contributions from second and third order reflections are even less significant. Therefore, it is highly justifiable to focus solely on the received power from the aligned portions of the RISs, while neglecting higher-order random reflections— a simplification that is numerically validated by the presented results.

\begin{figure}[t!]  
\begin{center}
\includegraphics[width=0.99\linewidth, draft=false]{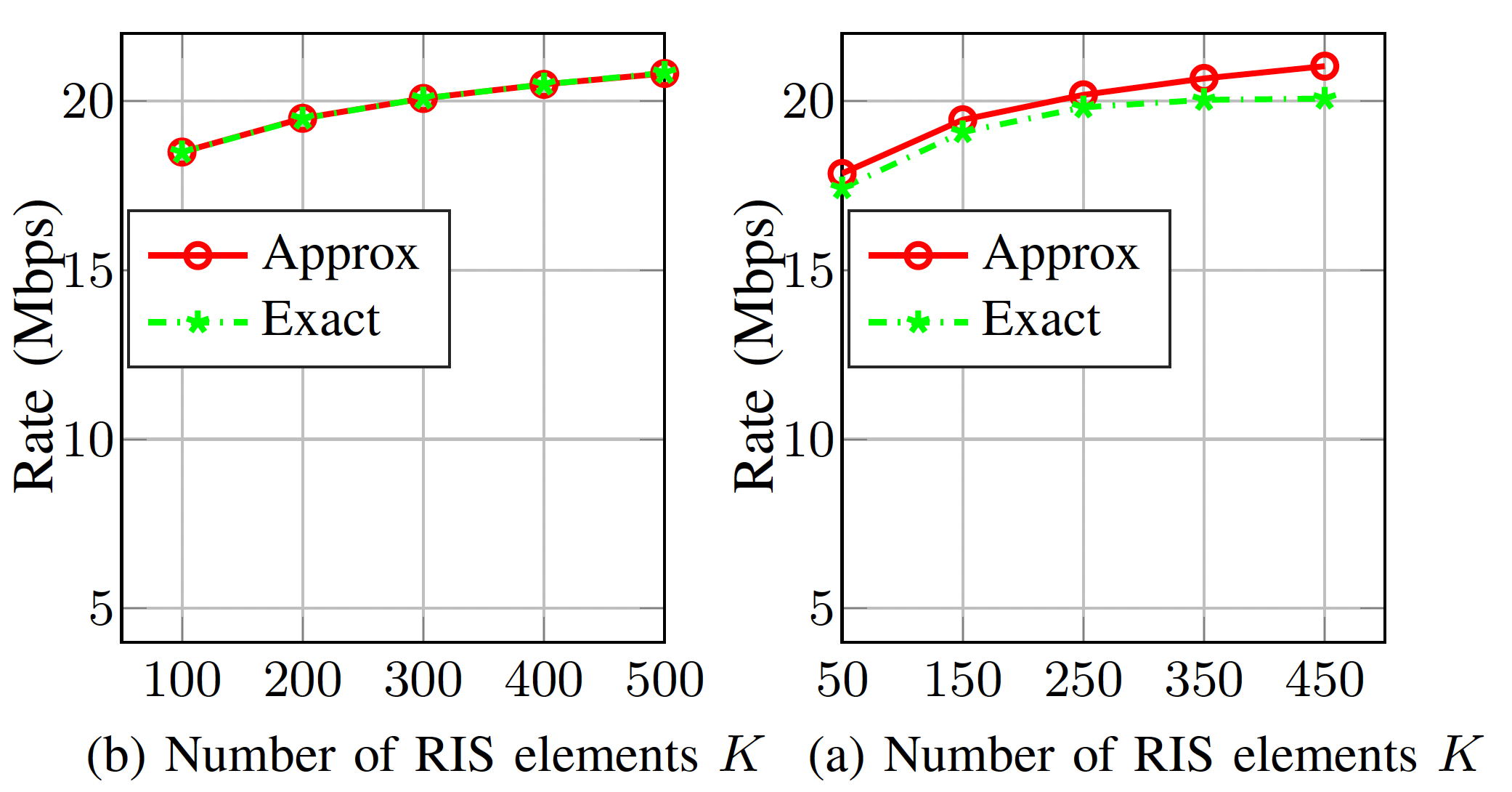}
\caption{Rate versus the number of elements $K$ for $U=4$, $U_r=2$, $A=2$, and $R=2$ for (a) perfect phase shift and (b) imperfect phase shift ($b=3$).}
   \label{fig3}
\end{center}
 
\end{figure}

 \begin{figure}[t!]  
\begin{center}
\includegraphics[width=0.99\linewidth, draft=false]{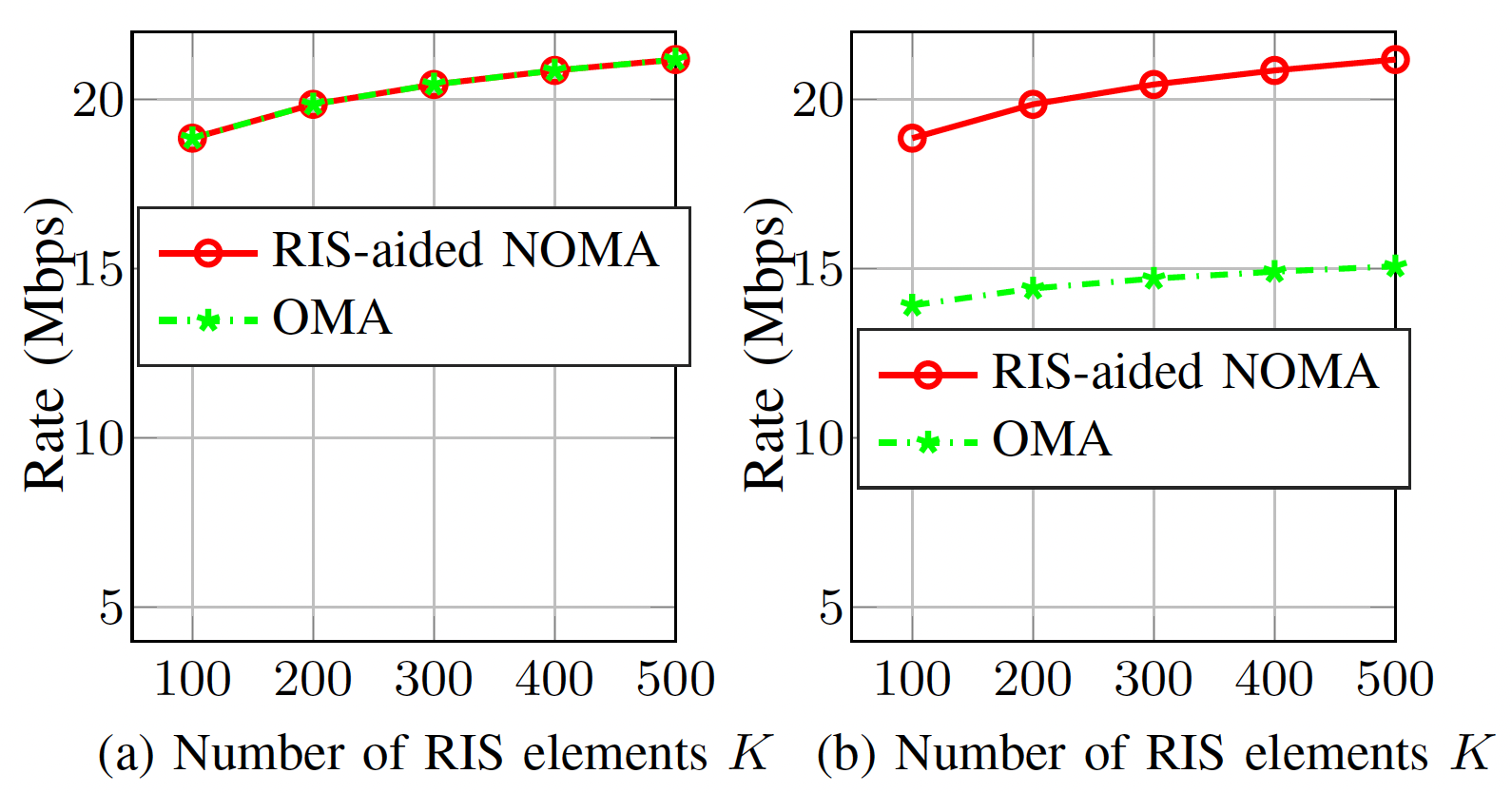}
\caption{Rate versus the number of elements $K$ for $U=4$, $A=2$, and $R=2$ for (a) $U_r=1$ and (b) $U_r=2$.}
   \label{fig4}
\end{center}
 
\end{figure}

Fig.~\ref{fig4} depicts the rates for both NOMA and orthogonal multiple access (OMA). In OMA, each associated UE (i.e., with $U_r=2$) uses half of the bandwidth.  It is evident that the rate increases linearly with the number of RIS elements. Furthermore, NOMA and OMA have the same performance when $U_r=1$ since both utilize the full RIS to reflect the signal of one UE. On the other hand, NOMA
consistently outperforms OMA in terms of rate across different numbers of RIS elements when $U_r=2$. This showcases the advantages of employing NOMA in RIS-aided wireless systems.

\subsection{Network Connectivity: Performance Analysis}
%\subsection{Network Connectivity: Benchmark Schemes}
For assessing the performance of the proposed RIS-aided UAV NOMA system, we compare it with the following schemes. 
\begin{itemize}

\item Each RIS creates one strong RIS-aided link, where the whole RIS is being utilized to narrow-beam the signal of the associated UE. Under this scenario, we implement \textbf{Convex Optimization} scheme by performing convex relaxation and SDP optimization for UE-RIS-UAV link selection. This scheme is widely used in \cite{4786516, saifglobecom_E, saifglobecom}.

   % \item \textbf{Matrix Perturbation:} This system uses perturbation heuristic solution that  adds $R$ edges one at a time. This scheme is near-optimal and proposed in \cite{saifglobecom_E}.

    %\item \textbf{Iterative R links:} We employ Algorithm \ref{alg_all} without  RIS partitioning, since a RIS utilizes all its elements to narrow-beam the signal of the associated UE to one UAV.

%\item \textbf{Random:} In this system, the selection of the UEs and the UAVs and the RIS partitioning are random. 

\item \textbf{Traditional UAV:} This system represents a traditional UAV communication system without RIS assistance, which is inspired by \cite{8292633}.

\item \textbf{Single RIS System:} In this system, a single and large RIS with $K_\text{single}=RK$ is considered to create $U_r$ links.

\end{itemize}

 \begin{figure}[t!]  
\begin{center}
\includegraphics[width=0.99\linewidth, draft=false]{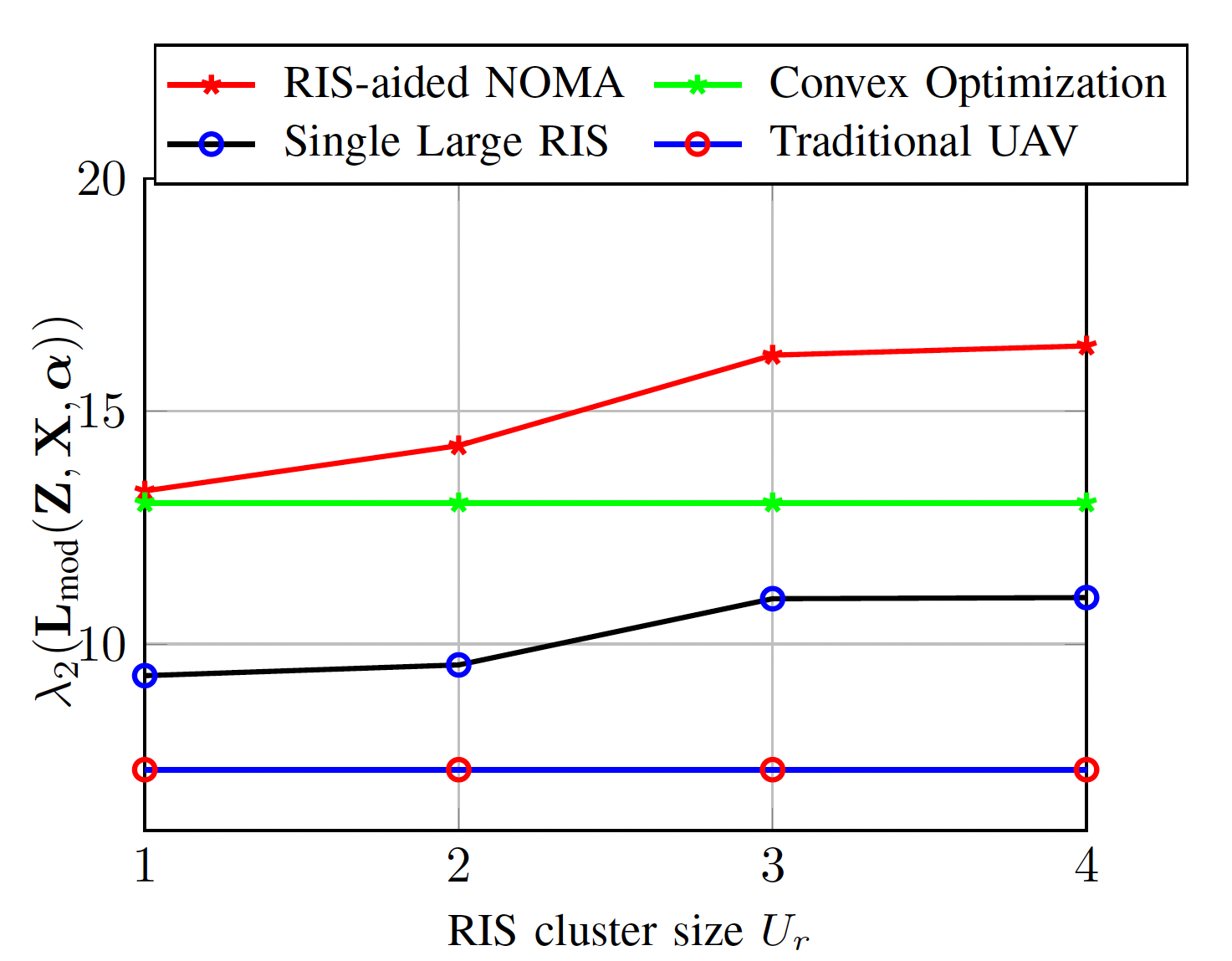}
\caption{Network connectivity versus RIS cluster size $U_r$ for $U=15$, $A=8$, $K=200$, and $R=3$.}
   \label{fig5a}
\end{center}
 
\end{figure}

%%%%%%%%%%%%%%%%%%%%%%%%%%%%%%
In Fig. \ref{fig5a}, we plot the average network connectivity versus the number of selected UEs $U_r$ (i.e., RIS cluster size) for a network composed of  $15$ UEs, $8$ UAVs, $3$ RISs, each has $200$ elements. We set $\gamma^{\text{UE}}_0= 84$ dB and $\gamma^{\text{UAV}}_0= 82$. As seen in Fig.~\ref{fig5a}, when optimizing for $U_r$, RIS-aided NOMA achieves the best performance across all values of $U_r$. More specifically, we can see that RIS-aided NOMA has the best performance, then convex, the single RIS, and finally, the traditional UAV. In particular, when $U_r$ is sufficiently small ($U_r=1$), the proposed and convex schemes have close performance since each RIS reflects the signal of an UE to only one UAV. In contrast, when $U_r$ is moderate ($U_r=2, 3$), the proposed scheme achieves  notable performance improvement (nearly $5$\% for $U_r=2$ and $12$\%  for $U_r=3$) over the convex scheme. This improvement is because  the multiple UE-RIS-UAV cascaded links increase as $U_r$ increases. This showcases the importance of judiciously optimizing RIS partitioning and NOMA UE clustering  in the sense that optimizing for $U_r$ did translate to higher network connectivity. However, dividing the RIS into multiple virtual partitions decreases the quality of the  RIS-aided links. This can be seen  at $U_r=4$, where each RIS is divided into $4$ partitions, the improvement does not increase much. The traditional UAV scheme exhibits consistent performance regardless of increasing  $U_r$, thus maintaining fixed network connectivity. Fig. \ref{fig5a} demonstrates that adding more cascaded RIS links does not significantly enhance connectivity once the network is well-connected, especially when the connections between the nodes are already dense.

%%%%%%
\begin{figure}[t!]  
\begin{center}
\includegraphics[width=0.99\linewidth, draft=false]{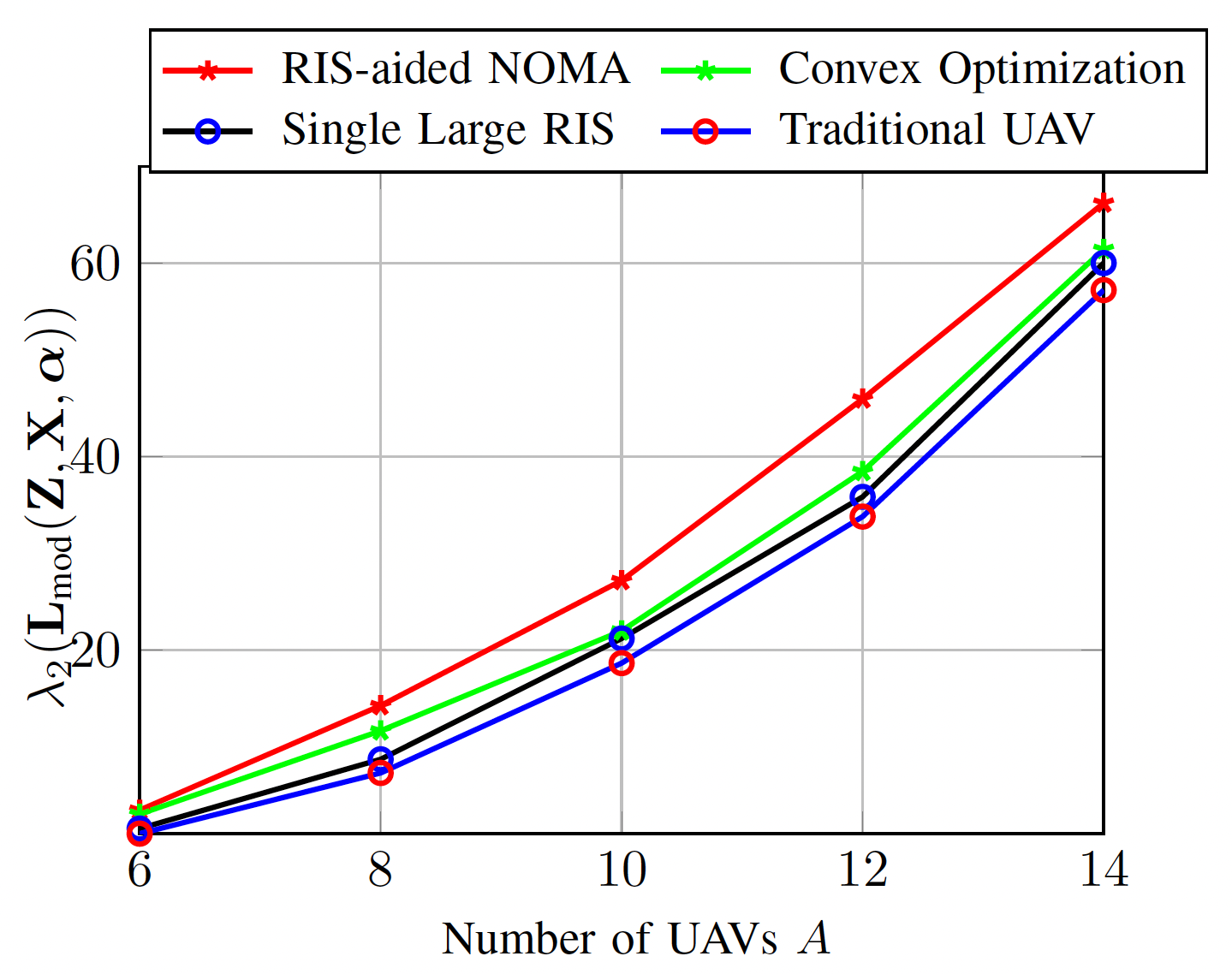}
\caption{Network connectivity versus number of UAVs $A$ for $U=15$, $U_r=3$, $K=200$, and $R=3$.}
   \label{fig5}
\end{center}
 
\end{figure}

\begin{figure}[t!]  
\begin{center}
\includegraphics[width=0.99\linewidth, draft=false]{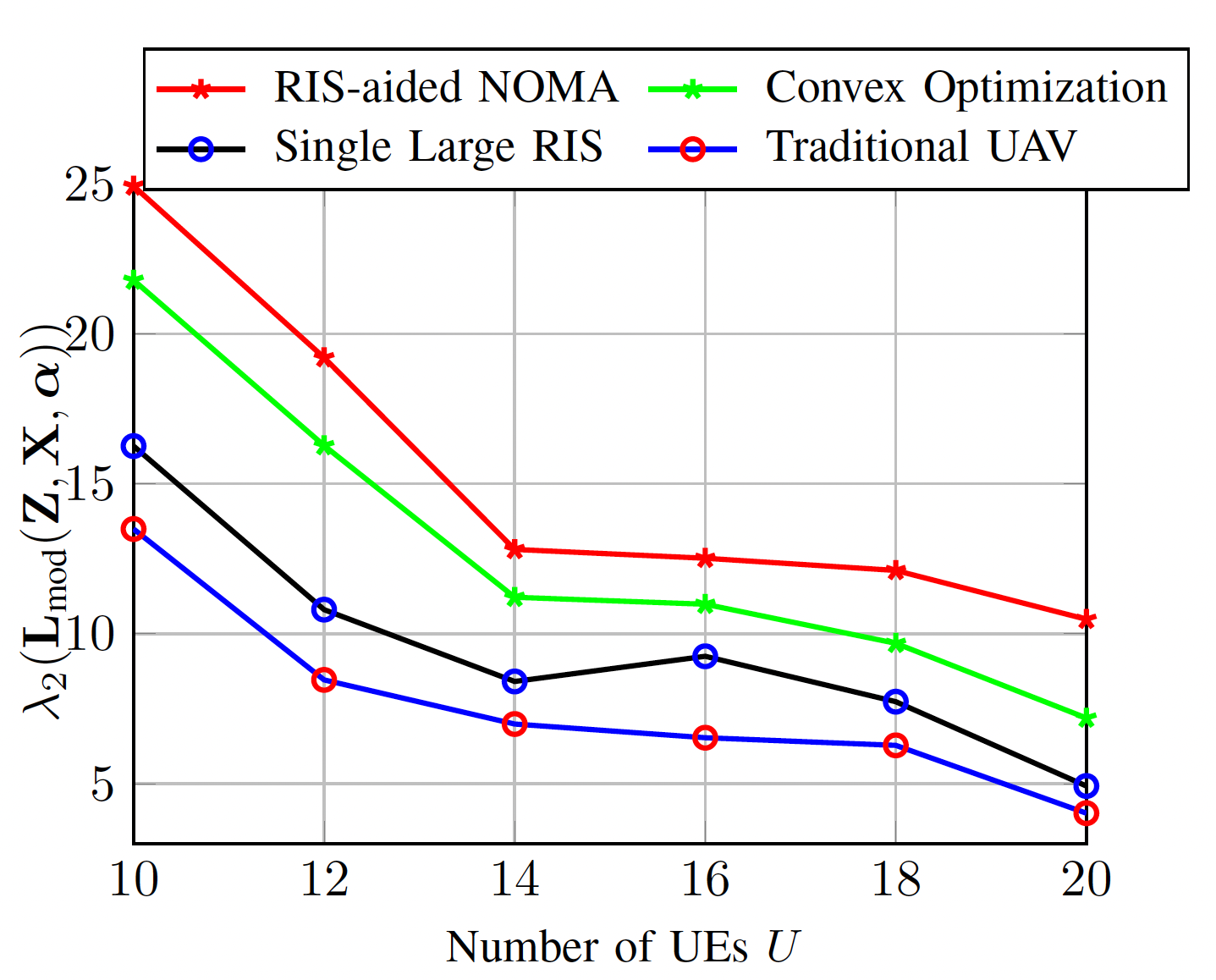}
\caption{Network connectivity versus number of UEs $U$ for $A=8$, $U_r=3$, $K=200$, and $R=3$.}
   \label{fig6}
\end{center}
 
\end{figure}
 
In Figs. \ref{fig5} and \ref{fig6}, we show the connectivity performance versus the number of UAVs $A$ and the number of UEs $U$, respectively. The numbers of RIS elements, RIS-aided links, and RISs are fixed to $200$, $3$, and $3$, respectively, while the number of UAVs $A$ is varied in the range of $[6, 14]$ with $U=15$ in Fig.~\ref{fig5} and the number of UEs $U$ is varied in the range of $[10, 20]$ with $A=8$ in Fig.~\ref{fig6}. From both figures, we can see that RIS-aided NOMA  consistently outperforms the other schemes. As the number of UAVs increases in Fig. \ref{fig5}, the connectivity of all  schemes significantly improves. This is because   most of $\text{UAV}_a \rightarrow\text{UAV}_{a'}$ link connections satisfy the $\gamma^{\text{UAV}}_0$, and adding more UAV nodes creates additional UAV-UAV and UE-UAV edges in $\mathcal G_\text{org}$. The enhanced connectivity as a result of the increased number of UAVs  facilitates more robust and reliable network connections. On the other hand, the connectivity performance of Fig. \ref{fig6} is slightly degraded for all the schemes, indicating that adding more UEs does not necessarily add more edges to the graph, as there are no direct connections between the UEs.

\begin{figure}[t!]  
\begin{center}
\includegraphics[width=0.99\linewidth, draft=false]{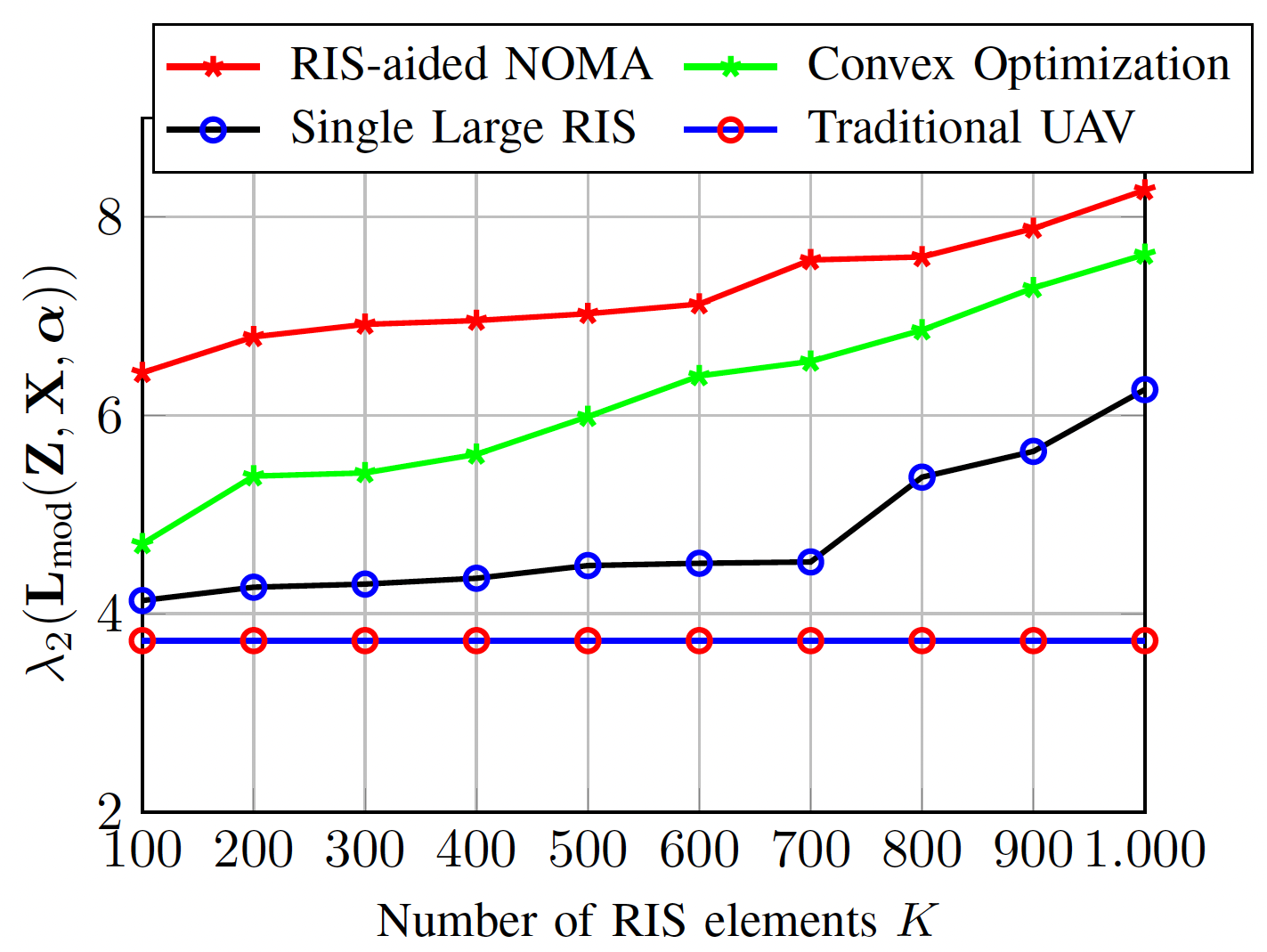}
\caption{Network connectivity versus number of RIS elements $K$ for $A=8$, $U_r=3$, $U=20$, and $R=3$.}
   \label{fig7}
\end{center}
 
\end{figure}

In Fig. \ref{fig7}, we show the network connectivity versus the number of RIS elements $K$ over a range of $[100, 1000]$, where $A=8$, $U_r=3$, $U=20$, and $R=3$.
As can be seen, the increase in the number of RIS elements increases connectivity, since  RISs with more elements can boost up the quality of the new $\text{UE}_{{u=\{1, \ldots, U_r\}}} \rightarrow \text{UAV}_{a}$  links. Specifically, since our RIS partitioning closed-form solution calculates $\alpha_{u=\{2, \ldots, U_r\}}$ that satisfies the minimum QoS of $\text{UE} \rightarrow \text{UAV}_{a}$ links, the remaining RIS elements, which increases as $K$ increases, are configured for  $\text{UE}_{u=1} \rightarrow \text{UAV}_a$ link to maximize its SINR to the maximum value. Thus, this increases the network connectivity. The curve for traditional UAV is maintained constant since it does not change with the RIS elements.

\begin{figure}[t!]  
\begin{center}
\includegraphics[width=0.99\linewidth, draft=false]{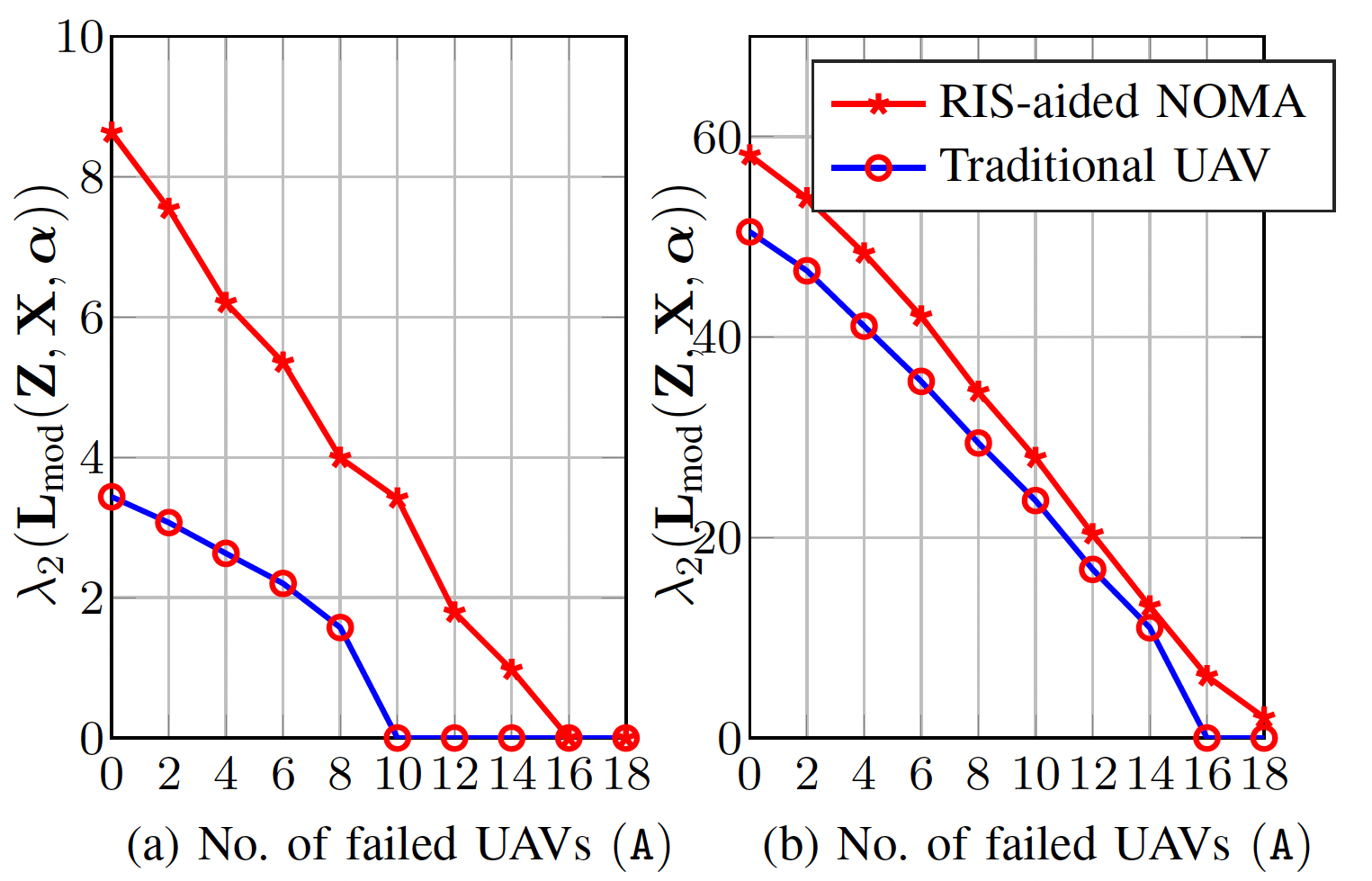}
\caption{Network connectivity versus number of failed UAVs $\mathtt A$ for (a) sparse network and (b) well-connected network.}
   \label{fig10}
\end{center}
 
\end{figure}

Given the designed RIS-aided NOMA system, Fig. \ref{fig10} shows the network connectivity of the proposed scheme compared to the traditional UAV scheme versus the number of failed UAVs $\mathtt A$ that are removed randomly in (a) sparse network with $\gamma^{\text{UE}}_0= 84$ dB and $\gamma^{\text{UAV}}_0= 90$ dB and (b) well-connected network with $\gamma^{\text{UE}}_0= 81$ dB and $\gamma^{\text{UAV}}_0= 80$ dB. We set $U=10$, $A=20$, $R=3$, $U
_r=4$, and $K=1000$.  When a sparse network is deployed, removing a few UAV nodes significantly reduces network connectivity. Consequently, when $\mathtt A=10$, the traditional UAV network becomes disconnected, thus has zero connectivity. The designed network, resulting from our proposed scheme, is more resilient against UAV failure, achieving zero connectivity at $\mathtt A=16$, thanks to NOMA UE clustering, RIS deployment and partitioning. We observe the same behavior for the well-connected network, but with a slower decrease in network connectivity to zero, i.e., at $\mathtt A=16$ for traditional UAV network and $\mathtt A>18$ for RIS-aided NOMA.

\begin{figure}[t!]  
\begin{center}
\includegraphics[width=0.99\linewidth, draft=false]{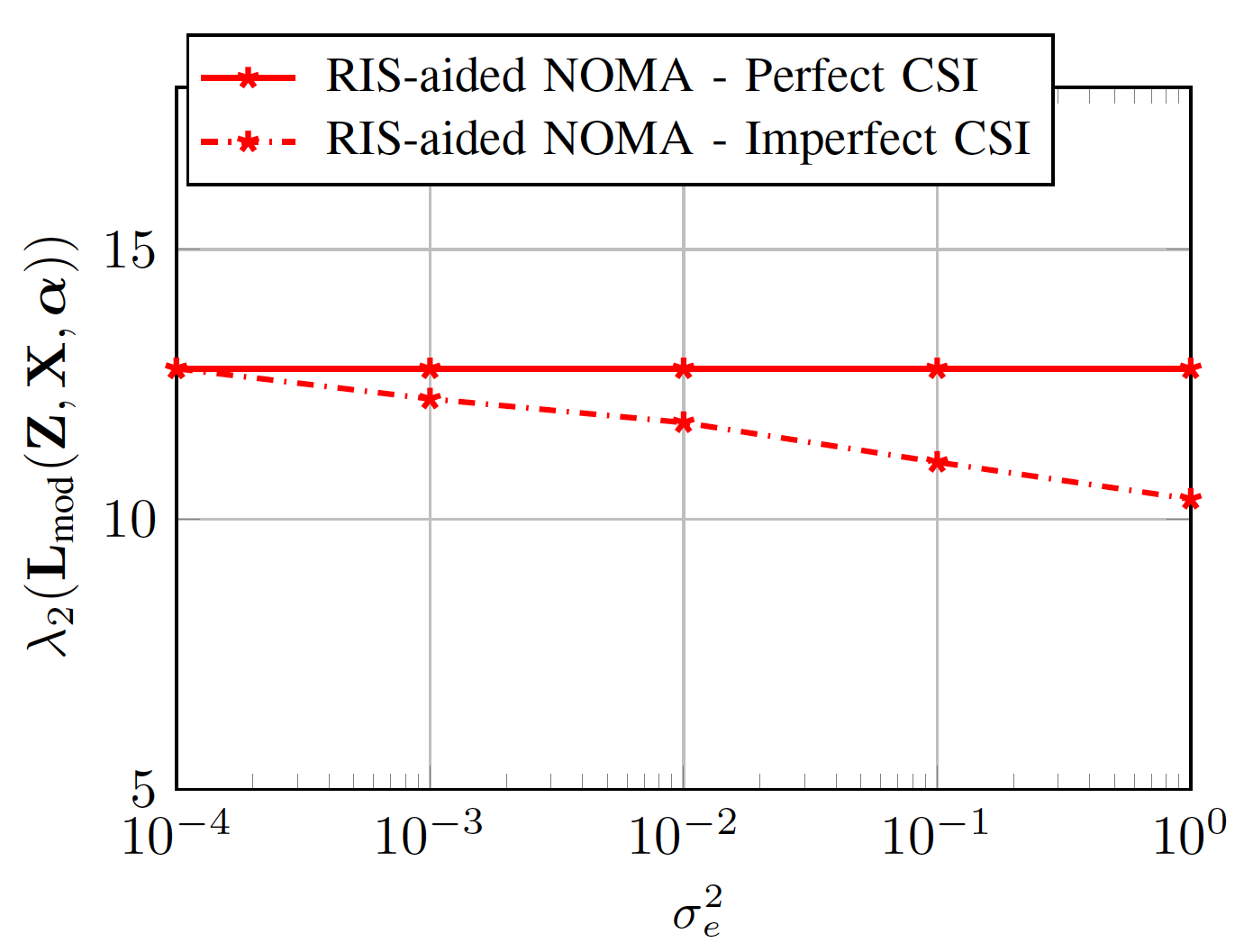}
\caption{Network connectivity versus $\sigma^2_e$ for $A=8$, $U_r=2$, $U=15$, $K=200$, and $R=3$.}
   \label{fig8}
\end{center}
 
\end{figure}

Next, we discuss the impact of imperfect CSI on the performance of the proposed RIS-aided NOMA scheme. We consider the imperfect CSI model in \cite{Imp3, Imp4} as follows. Generally, the channel
gain for the $i-j$ link can be modeled as $g_{i,j}=\sqrt{\beta_{i,j}}h_{i,j}$, where $h_{i,j}$ represents the small-scale fading and is expressed as $h_{i,j}=\hat{h}_{i,j}+h_{e, ij}$; where $\hat{h}_{i,j}$ is the estimated channel gain
that also follows the Nakagami-$f$ distribution with the corresponding parameters of the links UE-RIS, RIS-UAV, and UE-UAV, $h_{e, ij} \sim \mathcal {CN}(0, \sigma^2_{e, ij})$ denotes the estimation error, and $\beta_{i,j}$ is the large-scale model. Specifically, we denote the end-to-end CSI error of $\text{UE}_u \rightarrow \text{UAV}_a$ and $\text{UE}_u \xrightarrow{\text{RIS}_r} \text{UAV}_{a}$ links by $\sigma^2_{e, ua}$ and $\sigma^2_{e, ura}$, respectively. Referring to Section II, $\beta_{i,j}=\beta^\text{UK}_{u,a}$ for the $\text{UE}_u \rightarrow \text{UAV}_a$ link, $\beta_{i,j}=\Tilde{G}_{u,r}$ for $\text{UE}_u \rightarrow \text{RIS}_r$ link, and $\beta_{i,j}=\Tilde{H}_{r,a}$ for the $\text{RIS}_r \rightarrow \text{UAV}_a$ link. Similarly, $\hat{h}_{i,j}=h^\text{UK}_{u,a}$ for $\text{UE}_u \rightarrow \text{UAV}_a$ link, $\hat{h}_{i,j}=g_{u,r}^{(k)}$ for $\text{UE}_u \rightarrow \text{RIS}_r$ link with the $k$-th RIS element, and $\hat{h}_{i,j}=h_{r,a}^{(k)}$ for $\text{RIS}_r \rightarrow \text{UAV}_a$ link with the $k$-th RIS element. Therefore, the SINR at UAV$_a$  can be expressed as  $\gamma^{(r)}_{u,a}(\boldsymbol{\alpha}_r)=  \frac{\Tilde{\gamma}_u+(\alpha^{a}_{u,r})^2 K^2m\hat{\gamma}_u}{\sum\limits_{\substack{u\in \mathcal U_r \\ u'\neq u}}\bigg(\Tilde{\gamma}_{u'}+(\alpha^{a}_{u',r})^2 K^2m\hat{\gamma}_{u'}\bigg)+1+\sigma^2_{e, ua}\Tilde{\gamma}_u+\sigma^2_{e, ura}\hat{\gamma}_u}$

Fig.~\ref{fig8} illustrates the impact of imperfect CSI coefficients $\sigma^2_e$ on network connectivity. The scenario considered includes $A=8$, $U_r=2$, $U=15$, $K=200$, $R=3$, and $\beta_0=-20$ dB, with the imperfection model and related values taken from the literature \cite{Imp1, Imp2}.  We assume $\sigma^2_{e, ua}=\sigma^2_{e, ura}=\sigma^2_{e}$, and present the network connectivity for both perfect CSI ($\sigma^2_{e}=0$) and imperfect CSI as a function of $\sigma^2_{e}$. As can be seen in the figure, imperfect CSI introduces some performance degradation compared to the perfect CSI case, however, the overall impact remains moderate until $\sigma^2_e >10^0$, where a more noticeable impact is observed. In particular, the influence of imperfect CSI mainly affects the SINR of the RIS-aided links, which contributes to degradation of network connectivity. However, since network connectivity is influenced not only by the SINRs of the RIS-assisted links but also by the overall graph structure — including NOMA cluster formation, the reliable scheduling of UAVs, and  the addition of RIS-aided links — the impact of CSI errors on connectivity remains moderate.

\begin{table*}[t!]
    \renewcommand{\arraystretch}{0.9}
    \caption{Execution time (in seconds) and performance of the simulated schemes for varying numbers of nodes.}
    \label{table_2}
    \centering
    \begin{tabular}{|c||cc|cc|cc|cc|}
        \hline
        \multirow{3}{*}{\textbf{Schemes}} & \multicolumn{8}{c|}{\textbf{Number of Nodes (UEs \& UAVs)}} \\
        \cline{2-9}
        & \multicolumn{2}{c|}{\textbf{25}} & \multicolumn{2}{c|}{\textbf{35}} & \multicolumn{2}{c|}{\textbf{45}} & \multicolumn{2}{c|}{\textbf{55}} \\
        \cline{2-9}
        & \textbf{Time} & \textbf{Perf.} & \textbf{Time} & \textbf{Perf.} & \textbf{Time} & \textbf{Perf.} & \textbf{Time} & \textbf{Perf.} \\
        \hline \hline
        RIS-aided NOMA        & $0.148626$ & $30.0889$ & $0.203271$ & $68.9734$ & $0.223271$ & $88.9734$ & $0.253271$ & $100.4974$ \\
        Convex Optimization   & $2.570838$ & $25.6707$ & $5.901637$ & $60.4974$ & $10.329478$ & $78.4974$ & $15.329478$ & $86.4974$ \\
        Single Large RIS      & $0.030709$ & $22.3508$ & $0.064330$ & $58.5885$ & $0.038614$ & $72.4974$ & $0.048614$ & $78.4974$ \\
        Traditional UAV       & $0.030534$ & $19.5680$ & $0.034065$ & $55.6054$ & $0.032308$ & $65.6054$ & $0.034308$ & $75.6054$ \\
        \hline
    \end{tabular}
\end{table*}
 
Finally, \tref{table_2} studies  the run time (in seconds for one algorithm iteration) of MATLAB for all the schemes and their corresponding connectivity performance. The simulations to calculate the running time are carried in MATLAB on a macOS 12.5.1 iMac Apple M1 chip and 8 GB  memory. We consider a network setup of $3$
RISs with different number of nodes (UEs and UAVs). \tref{table_2} shows that the convex scheme requires high computation time than all the other solutions. This is due to the fact that this scheme solves the SDP optimization of all candidate RIS-aided links in the network. On the other hand, our proposed scheme is efficient in terms of low computation time and good network connectivity performance. Thus, our proposed scheme can be executed quickly, making it a preferred method for application in UAV networks.

\section{Conclusion}\label{C}
In this paper, we investigate the connectivity of RIS-assisted UAV NOMA networks from a link-layer perspective by leveraging UE and UAV clustering, along with RIS partitioning strategies. UEs are first clustered with RISs based on their channel gains, and UAVs are then associated with these clusters according to a reliability metric that captures the criticality of each UAV. We derive closed-form expressions for optimal RIS element allocation to enhance signal coverage and connectivity. Extensive simulations demonstrate that the proposed scheme significantly outperforms three benchmark systems, validating the effectiveness of the clustering and RIS partitioning design in enhancing network performance. For future work, optimizing UAV trajectories could further improve network adaptability and efficiency.


\begin{thebibliography}{00}
\bibitem{Cisco}
Cisco. (2020). Cisco Annual Internet Report (2018–2023). [Online].
Available: https://www.cisco.com/c/en/us/solutions/collateral/executiveperspectives/
annual-internet-report/white-paper-c11-741490.html


 \bibitem{RIS_Mohanad} 
M. Obeed and A. Chaaban, ``Joint beamforming design for multiuser MISO downlink aided by a reconfigurable intelligent surface and a relay," in \emph{IEEE Trans. on Wireless Commun.,} vol. 21, no. 10, pp. 8216-8229, Oct. 2022.


\bibitem{9293155} 
Z. Wei, Y. Cai, Z. Sun, D. Wing Kwan Ng, J. Yuan, M. Zhou, and L. Sun, ``Sum-rate maximization for IRS-assisted UAV OFDMA communication systems," in \emph{IEEE Trans. on Wireless Commun.,} vol. 20, no. 4, pp. 2530-2550, Apr. 2021.




\bibitem{saifglobecom_E} 
M. Saif, M. Javad-Kalbasi, and S. Valaee, ``Effectiveness of reconfigurable intelligent surfaces to enhance connectivity in UAV networks",  \emph{IEEE
Trans. on Wireless Commun.,} vol. 23, no. 12, pp. 18757-18773, Dec. 2024. 

 
\bibitem{saifglobecom} 
M. Saif, M. Javad-Kalbasi, and S. Valaee, ``Maximizing network  connectivity for UAV communications via reconfigurable intelligent surfaces", \emph{2023 IEEE Global Commun. Conference}, Kuala Lumpur, Malaysia, 2023, pp. 6395-6400.

\bibitem{Javad_globecom2023} 
M. Javad-Kalbasi, M. Saif, and S. Valaee, ``Energy efficient communications in RIS-assisted UAV networks based on genetic algorithm", \emph{2023 IEEE Global Commun. Conference}, Kuala Lumpur, Malaysia, 2023, pp. 5901-5906.

\bibitem{PLS} 
S. Arzykulov, A. Celik, G. Nauryzbayev, and A. M. Eltawil, ``Aerial RIS-aided physical layer security: Optimal deployment and partitioning," in \emph{IEEE Trans. on Cognitive Commun. and Net.}, vol. 10, no. 5, pp. 1867-1882, Oct. 2024.


\bibitem{M}
M. Ammous and S. Valaee, ``Cooperative positioning with the aid of reconfigurable intelligent surfaces and zero access points," \emph{2022 IEEE 96th Vehicular Tech. Conf. (VTC2022-Fall), London, UK,} pp. 1-5.

\bibitem{Ali} 
 A. Parchekani and S. Valaee, ``Reconfigurable intelligent surface assisted sensing and localization using the Swendsen-Wang and evolutionary algorithms," \emph{2023 IEEE 98th Vehicular Technology Conf. (VTC2023-Fall),} Hong Kong, 2023, pp. 1-6.

 %%10
\bibitem{Mu} 
H. Chen, H. Kim, M. Ammous, G. Seco-Granados, G. C. Alexandropoulos, S. Valaee, and H. Wymeersch, ``RISs and sidelink communications in smart cities: The key to seamless localization and sensing," in \emph{IEEE Commun. Magazine,} vol. 61, no. 8, pp. 140-146, Aug. 2023.

 



%%%%%%%%%%%%%%
\bibitem{UAV_economy}
S. Zhang, Y. Zeng, and R. Zhang, ``Cellular-enabled UAV communication: Aconnectivity-constrained trajectory optimization perspective,” \emph{IEEE Trans. Commun.,} vol. 67, no. 3, pp. 2580-2604, Mar. 2019.
 




\bibitem{Rui_Zhang_UAV}
S. Zhang, Y. Zeng and R. Zhang, ``Cellular-enabled UAV communication: A connectivity-constrained trajectory optimization perspective,''
\textit{IEEE Trans. Commun.}, vol. 67, no. 3, pp. 2580–2604, Mar. 2019.

\bibitem{3GPP}
\textit{Technical Speciﬁcation Group Services and System Aspects; Unmanned Aerial System (UAS) Support in 3GPP (Release 17)}, Standard 3GPP TS 22.125, 3rd Generation Partnership Project, Dec. 2019. [Online] Available: https://www.3gpp.org/ftp/Specs/archive/22 series/22.125/.

\bibitem{8292633}
M. A. Abdel-Malek, A. S. Ibrahim, and M. Mokhtar, ``Optimum UAV positioning for better coverage-connectivity trade-off," \emph{2017 IEEE 28th Annual Intern. Symposium on Personal, Indoor, and Mobile Radio Commun. (PIMRC),} Montreal, QC, Canada, 2017, pp. 1-5.


   \bibitem{4786516}
A. S. Ibrahim, K. G. Seddik, and K. J. R. Liu, ``Connectivity-aware network maintenance and repair via relays deployment," in \emph{IEEE Trans. on Wireless Commun.,} vol. 8, no. 1, pp. 356-366, Jan. 2009.

\bibitem{4657335}
 C. Pandana and K. J. R. Liu, ``Robust connectivity-aware energy-efficient routing for wireless sensor networks," in \emph{IEEE Trans. on Wireless Commun.,} vol. 7, no. 10, pp. 3904-3916, Oct. 2008.

 \bibitem{H}
H. Dahrouj, A. Douik, F. Rayal, T. Y. Al-Naffouri, and M.-S. Alouini, ``Cost-effective hybrid RF/FSO backhaul solution for next generation wireless systems," in \emph{IEEE Wireless Commun.,} vol. 22, no. 5, pp. 98-104, Oct. 2015.

\bibitem{4A}
Q. Wu, S. Zhang, B. Zheng, C. You, and R. Zhang, ``Intelligent reflecting surface-aided wireless communications: A tutorial,” \emph{IEEE
Trans. Commun.,} vol. 69, no. 5, pp. 3313-3351, May 2021.

\bibitem{6A}
C. Huang, A. Zappone, G. C. Alexandropoulos, M. Debbah, and
C. Yuen, ``Reconfigurable intelligent surfaces for energy efficiency in
wireless communication,” \emph{IEEE Trans. Wireless Commun.,} vol. 18, no. 8,
pp. 4157-4170, Aug. 2019.

 \bibitem{10A}
Q. Tao, J. Wang, and C. Zhong, ``Performance analysis of intelligent reflecting surface aided communication systems,” \emph{IEEE Commun. Lett.,} vol. 24, no. 11, pp. 2464-2468, Nov. 2020.

\bibitem{5A} 
S. Basharat, S. A. Hassan, H. Pervaiz, A. Mahmood, Z. Ding, and M. Gidlund, ``Reconfigurable intelligent surfaces: Potentials, applications,
and challenges for 6G wireless networks,” \emph{IEEE Wireless Commun.,}
vol. 28, no. 6, pp. 184-191, Dec. 2021.



\bibitem{CON}
S. Boyd, ``Convex optimization of graph laplacian eigenvalues," in \emph{Proc.
International Congress of Mathematicians,} vol. 3, pp. 1311-1319, 2006.

\bibitem{aydin} 
K. Weinberger, R.-J. Reifert, A. Sezgin, and E. Basar, ``RIS-enhanced resilience in cell-free MIMO," WSA \& SCC 2023; \emph{26th International ITG Workshop on Smart Antennas and 13th Conference on Systems, Communications, and Coding,} Braunschweig, Germany, 2023, pp. 1-6.


\bibitem{saifVTC}
M. Saif and S. Valaee, ``Improving connectivity of RIS-assisted UAV networks using RIS partitioning and deployment,”  \emph{2024 IEEE 100th Vehicular Tech. Conf. (VTC2024-Fall),} DC, Washington, USA, pp. 1-6.
 
\bibitem{saifTCOM}
M. Saif and S. Valaee, ``RIS alignment via virtual partitioning for resilient uplink multi-RIS-assisted UAV communications,”  \emph{EEE Trans. on Commun.,} early access, 2025.
%%%%%%%%%%%%%%%%%%
\bibitem{NOMA-1}
I. Budhiraja et al., ``A Systematic review on NOMA variants for 5G and beyond,” in \emph{IEEE Access,} vol. 9, pp. 85573-85644, 2021.
 

\bibitem{NOMA-2}
M. S. Al-Abiad, M. Z. Hassan and M. J. Hossain, ``Energy-efficient resource allocation for federated learning in NOMA-enabled and relay-assisted internet of things networks," in \emph{IEEE Internet of Things Journal,} vol. 9, no. 24, pp. 24736-24753, 15 Dec.15, 2022.

\bibitem{NOMA-3}
J. Zuo, Y. Liu, Z. Qin, and N. Al-Dhahir, ``Resource allocation in
intelligent reflecting surface assisted NOMA systems,” \emph{IEEE Trans.
Commun.,} vol. 68, no. 11, pp. 7170–7183, Nov. 2020.

 \bibitem{NOMA-4}
A. Celik, ``Grant-free NOMA: A low-complexity power control through
user clustering,” \emph{Sensors,} vol. 23, no. 19, p. 8245, Oct. 2023.

%  \bibitem{NOMA-5}
% T. Hou, Y. Liu, Z. Song, X. Sun, Y. Chen, and L. Hanzo, ``Reconfigurable intelligent surface aided NOMA networks,” \emph{IEEE J. Sel. Areas Commun.,} vol. 38, no. 11, pp. 2575-2588, Nov. 2020.


% \bibitem{NOMA-6}
%  Y. Cheng, K. H. Li, Y. Liu, K. C. Teh, and H. V. Poor, ``Downlink and uplink intelligent reflecting surface aided networks: NOMA and
% OMA,” \emph{IEEE Trans. Wireless Commun.,} vol. 20, no. 6, pp. 3988-4000, Jun. 2021.


% \bibitem{NOMA-7}
%  M. H. N. Shaikh, V. A. Bohara, A. Srivastava, and G. Ghatak,
% ``A downlink RIS-aided NOMA system with hardware impairments:
% Performance characterization and analysis,” \emph{IEEE Open J. Signal Process.,} vol. 3, pp. 288-305, 2022.


 

 
 \bibitem{NOMA}
M. H. N. Shaikh, A. Celik, A. M. Eltawil, and G. Nauryzbayev, ``Grant-free NOMA through optimal partitioning and cluster assignment in STAR-RIS networks," in \emph{IEEE Trans.on Wireless Commun.,} vol. 23, no. 8, pp. 10166-10181, Aug. 2024. 

 \bibitem{optimal}
M. Makin, S. Arzykulov, A. Celik, A. M. Eltawil and G. Nauryzbayev, ``Optimal RIS partitioning and power control for bidirectional NOMA networks," in \emph{IEEE Trans. on Wireless Commun.,} vol. 23, no. 4, pp. 3175-3189, April 2024.

\bibitem{newnewnew}
E. Arslan, F. Kilinc, S. Arzykulov, A. Tugberk Dogukan, A. Celik, E. Basar,
and A. M. Eltawil, ``Reconfigurable intelligent surface enabled over-the-air uplink NOMA,” in \emph{IEEE Trans. on Green Commun. and Networking,} vol. 7, no. 2, pp. 814-826, Jun. 2023.

\bibitem{9A}
D. Selimis, K. P. Peppas, G. C. Alexandropoulos, and F. I. Lazarakis, ``On the performance analysis of RIS-empowered communications over Nakagami$-m$ fading,” \emph{IEEE Commun. Lett.,} vol. 25, no. 7, pp. 2191-2195, Jul. 2021.

\bibitem{CSI-1}
A. Abdallah, A. Celik, M. M. Mansour, and A. M. Eltawil, ``RIS-aided
mmWave MIMO channel estimation using deep learning and compressive sensing,” \emph{IEEE Trans. Wireless Commun.,} vol. 22, no. 5,
pp. 3503–3521, May 2023.


\bibitem{Ednew}
Q. Wu and R. Zhang, ``Intelligent reflecting surface enhanced wireless network via joint active and passive beamforming," in \emph{IEEE Trans. on Wireless Communications,} vol. 18, no. 11, pp. 5394-5409, Nov. 2019.




\bibitem{38}
S. Arzykulov, G. Nauryzbayev, A. Celik, and A. M. Eltawil, ``RIS-assisted full-duplex relay systems,” \emph{IEEE Systems Journal,} vol. 16, no. 4, pp. 5729–5740, 2022.


 \bibitem{new}
M. Fiedler, ``Algebraic connectivity of graphs," Czechoslovak Mathematical J., vol. 23, pp. 298-305, 1973. 



 \bibitem{comp} 
F. Kilinc, R. A. Tasci, A. Celik, A. Abdallah, A. M. Eltawil, and E. Basar, ``RIS-assisted grant-free NOMA: User pairing, RIS assignment, and phase shift alignment,” \emph{IEEE Trans. Cog. Commun. Netw.,} vol. 9, no. 5, pp. 1257-1270, Oct. 2023.






%\bibitem{4177113}
%A. Ghosh and S. Boyd, ``Growing well-connected graphs," \emph{Proceedings of the 45th IEEE Conference on Decision and Control,} San Diego, CA, USA, 2006, pp. 6605-6611.




 \bibitem{44A}
E. Bjornson, O. Ozdogan, and E. G. Larsson, ``Intelligent reflecting surface versus decode-and-forward: How large surfaces are needed to beat relaying?” \emph{IEEE Wireless Commun. Lett.,} vol. 9, no. 2, pp. 244–248, 2020.

 \bibitem{45A}
J. Lyu and R. Zhang, ``Spatial throughput characterization for intelligent reflecting surface aided multiuser system,” \emph{IEEE Wireless Communications Letters,} vol. 9, no. 6, pp. 834–838, 2020.

  \bibitem{Urban}
K. Haneda \textit{et al.}, ``5G 3GPP-like channel models for outdoor urban microcellular and macrocellular environments," \emph{2016 IEEE 83rd Vehicular Technology Conference (VTC Spring),} Nanjing, China, 2016, pp. 1-7.

 











\bibitem{Imp3}
Z. Mai \emph{et al.,} ``A {UAV} air-to-ground channel estimation algorithm based on deep learning,” \emph{Wireless Personal Communications,} vol. 124, no. 3,
%pp. 2247-2260, 2012.
  
\bibitem{Imp4}
D. Fan \emph{et al.,} ``Channel estimation and self-positioning for UAV swarm," in \emph{IEEE Trans. on Communications,} vol. 67, no. 11, pp. 7994-8007, Nov. 2019.

 \bibitem{Imp1}
 S. K. Singh, K. Agrawal, K. Singh, C. -P. Li, and Z. Ding, ``NOMA enhanced hybrid RIS-UAV-assisted full-duplex communication system with imperfect SIC and CSI," in \emph{IEEE Trans. on Communications,} vol. 70, no. 11, pp. 7609-7627, November 2022.

\bibitem{Imp2}
 F. Kara and H. Kaya, ``Improved user fairness in decode-forward relaying non-orthogonal multiple access schemes with imperfect SIC and CSI,” \emph{IEEE Access,} vol. 8, pp. 97540–97556, 2020.
\end{thebibliography}
\end{document}